\documentclass[
final
 , nomarks
]{dmtcs-episciences}

\usepackage[utf8]{inputenc}
\usepackage{amsmath,amssymb, amsthm}
\usepackage{hyperref}
\usepackage{subcaption}

\usepackage{tikz}
\usetikzlibrary{arrows, automata, shapes, positioning}
\usetikzlibrary{arrows.meta}
\usetikzlibrary{decorations.pathmorphing,patterns,decorations.pathreplacing}
\usetikzlibrary{calc}

\theoremstyle{plain}
\newtheorem{theorem}{Theorem}
\newtheorem{lemma}[theorem]{Lemma}
\newtheorem{corollary}[theorem]{Corollary}
\newtheorem{prop}[theorem]{Proposition}
\theoremstyle{definition}
\newtheorem{definition}[theorem]{Definition}
\newtheorem{remark}[theorem]{Remark}
\newtheorem{example}[theorem]{Example}

\DeclareMathOperator{\rep}{rep}
\DeclareMathOperator{\val}{val}
\DeclareMathOperator{\dom}{dom}

\author{Michel Rigo\affiliationmark{1}\thanks{The first author dedicates this paper to the memory of his grandmother Marie Wuidar (1923--2020).}
  \and Manon Stipulanti\affiliationmark{3}\thanks{The second author is supported by the FNRS Research grant 1.B.397.20F.}}
\title{Automatic sequences: from rational bases to trees}

\affiliation{
  Department of Mathematics, University of Li\`ege,  Belgium}
\received{2021-09-07}
\revised{2022-02-28}
\accepted{2022-05-31}

\begin{document}

\publicationdetails{24}{2022}{1}{25}{8455}
\maketitle
\begin{abstract}
The $n$th term of an automatic sequence is the output of a deterministic finite automaton fed with the representation of $n$ in a suitable numeration system.
In this paper, instead of considering automatic sequences built on a numeration system with a regular numeration language, we consider those built on languages associated with trees having periodic labeled signatures and, in particular, rational base numeration systems.  We obtain two main characterizations of these sequences. The first one is concerned with $r$-block substitutions where $r$ morphisms are applied periodically. In particular, we provide examples of such sequences that are not morphic. The second characterization involves the factors, or subtrees of finite height, of the tree associated with the numeration system and decorated by the terms of the sequence. 
\end{abstract}

\section{Introduction}

Motivated by a question of Mahler in number theory, the introduction of rational base numeration systems has brought to light a family of formal languages with a rich combinatorial structure \cite{Akiyama--Frougny-Sakarovitch-2008}. In particular, the generation of infinite trees with a periodic signature has emerged \cite{Marsault--Sakarovitch-1,Marsault--Sakarovitch-2,Marsault--Sakarovitch-2016,Marsault--Sakarovitch-2017}. Marsault and Sakarovitch very quickly linked the enumeration of the vertices of such trees (called breadth-first serialization) to the concept of abstract numeration system built on the corresponding prefix-closed language: the traversal of the tree is exactly the radix enumeration of the words of the language. In this paper, we study automatic sequences associated with that type of numeration systems. In particular, in the rational base $\frac{p}{q}$, a sequence is $\frac{p}{q}$-automatic if its $n$th term is obtained as the output of a DFAO fed with the base-$\frac{p}{q}$ representation of $n$. Thanks to a result of Lepist\"o \cite{Lepisto} on factor complexity, we observe that we can get sequences that are not morphic.

We obtain several characterizations of these sequences. The first one boils down to translating Cobham's theorem from 1972 into this setting. In Section~\ref{sec:cobham}, we show that any automatic sequence built on a tree language with a purely periodic labeled signature is the image under a coding of an alternating fixed point of uniform morphisms not necessarily of the same length. If all the morphisms had the same length, as observed in \cite{Endrullis}, we would only get classical $k$-automatic sequences. As a consequence, in the rational base $\frac{p}{q}$, if a sequence is $\frac{p}{q}$-automatic, then it is the image under a coding of a fixed point of a $q$-block substitution whose images all have length $p$. In the literature, these substitutions are also called PD0L where a periodic control is applied --- $q$ different morphisms are applied depending on the index of the considered letter modulo~$q$. 

On the other hand, Sturmian trees as studied in \cite{sturm} also have a rich combinatorial structure where subtrees play a special role analogous to factors occurring in infinite words. In Section~\ref{sec:tree}, we discuss about the factors, i.e., subtrees of finite height, that may appear in the tree whose paths from the root are labeled by the words of the numeration language and whose vertices are colored according to the sequence of interest. Related to the $k$-kernel of a sequence, we obtain a new characterization of the classical $k$-automatic sequences: a sequence $\mathbf{x}$ is $k$-automatic if and only if the labeled tree of the base-$k$ numeration system decorated by $\mathbf{x}$ is rational, i.e., it has finitely many infinite subtrees. For numeration systems built on a regular language, the function counting the number of decorated subtrees of height $n$ is bounded, and we get a similar result. This is not the case in the more general setting of rational base numeration systems. Nevertheless, we obtain sufficient conditions for a sequence to be $\frac{p}{q}$-automatic in terms of the number of extensions a subtree may have.

This paper is organized as follows. In Section~\ref{sec:prelim}, we recall basic definitions about abstract numeration systems, tree languages, rational base numeration systems, and alternating morphisms. 
In Section~\ref{sec:ASRB}, we give some examples of the automatic sequences that we will consider. The parity of the sum-of-digits in base $\frac{3}{2}$ is such an example.
In Section~\ref{sec:cobham}, Cobham's theorem is adapted to the case of automatic sequences built on tree languages with a periodic labeled signature in Theorem~\ref{thm:cobham} (so, in particular, to the rational base numeration systems in Corollary~\ref{cor:cas32}).
In Section~\ref{sec:tree}, we decorate the nodes of the tree associated with the language of a rational base numeration system with the elements of a sequence taking finitely many values. Under some mild assumption (always satisfied when distinct states of the deterministic finite automaton with output producing the sequence have distinct output), we obtain a characterization of $\frac{p}{q}$-automatic sequences in terms of the number of extensions for subtrees of some finite height occurring in the decorated tree. 
In Section~\ref{sec:stability}, we review some usual closure properties of $\frac{p}{q}$-automatic sequences.

\section{Preliminaries}\label{sec:prelim}

We make use of common notions in combinatorics on words and usual definitions from automata theory.
General references are~\cite{Lothaire,Shallit2009}. Let us also mention \cite{Allouche--Shallit-2003,BertheRigo} for references about automatic sequences and numeration systems. 
In particular, an \emph{alphabet} is a finite set of elements, which are themselves called the \emph{letters}.
For an alphabet $A$, we use $A^*$ (resp., $A^+$) for the set of finite words (resp., non-empty finite words) over $A$.
We let $\varepsilon$ denote the empty word.
So $A^+ = A^*\setminus\{\varepsilon\}$.
For a finite word $w$, we let $|w|$ denote its length.
For each $i\in\{0,\ldots,|w|-1\}$, we let $w_i$ denote the $i$th letter of $w$ (and we thus start indexing letters at $0$). 
For a word $w=w_0\cdots w_{|w|-1}$, we let $w^R$ denote its \emph{mirror} or \emph{reversal} $w_{|w|-1} \cdots w_0$.

\subsection{Abstract numeration systems}

When dealing with abstract numeration systems, it is usually assumed that the language of the numeration system is regular. However the main feature is that words are enumerated by radix order (also called genealogical order: words are first ordered by increasing length and words of the same length are ordered by lexicographical order). The generalization of abstract numeration systems to context-free languages was, for instance, considered in \cite{CLR}. Rational base numeration systems discussed below in Section~\ref{subsec:rational bases} are also abstract numeration systems built on non-regular languages.

\begin{definition}
  An {\em abstract numeration system} (or {\em ANS} for short) is a
  triple $\mathcal{S}=(L,A,<)$ where $L$ is an infinite language over
  a totally ordered (finite) alphabet ${(A,<)}$. 
 We say that $L$ is the \emph{numeration language}.
  The map
  $\rep_\mathcal{S}:\mathbb{N}\to L$ is the one-to-one correspondence mapping
  $n\in\mathbb{N}$ onto the $(n+1)$st word in the radix ordered language $L$,
  which is then called the {\em $\mathcal{S}$-representation} of $n$. The
  $\mathcal{S}$-representation of $0$ is the first word in $L$. The
  inverse map is denoted by $\val_\mathcal{S}:L\to\mathbb{N}$. For any word $w$ in $L$, $\val_\mathcal{S}(w)$ is its {\em
    $\mathcal{S}$-numerical value}.
\end{definition}

Positional numeration systems, such as integer base numeration systems, the Fibonacci numeration system, and Pisot numeration systems, are based on the greediness of the representations (computed through a greedy algorithm where at each step one subtracts, by Euclidean division, the largest available term of the sequence from the remaining part to be represented~\cite{Fraenkel1985}). They all share the following property: $m<n$ if and only if $\rep(m)$ is less than $\rep(n)$ for the radix order. These numeration systems are thus ANS. 
As a non-standard example of ANS, consider the language $a^*b^*$ over $\{a,b\}$ and assume that $a<b$. Let $\mathcal{S}=(a^*b^*,\{a,b\},<)$. The first few words in the numeration language are $\varepsilon,a,b,aa,ab,bb,\ldots$. For instance, $\rep_\mathcal{S}(3)=aa$ and $\rep_\mathcal{S}(5)=bb$. One can show that $\val_\mathcal{S}(a^pb^q)=\frac{(p+q)(p+q+1)}{2}+q$. For details, we refer the reader to \cite{LR} or \cite{Rigo}.

In the next definition, we assume that most significant digits are read first. This is not a real restriction (see Section~\ref{sec:stability}).

\begin{definition}
  Let $\mathcal{S}=(L,A,<)$ be an abstract numeration system and let $B$ be a finite alphabet.  An
  infinite word $\mathbf{x}=x_0x_1x_2\cdots\in B^\mathbb{N}$ is {\em
    $\mathcal{S}$-automatic} if there exists a deterministic finite automaton with output (DFAO for short)
  $\mathcal{A}=(Q,q_0,A,\delta,\mu:Q\to B)$ such that
  $x_n=\mu(\delta(q_0,\rep_\mathcal{S}(n)))$ for all $n\ge 0$.
\end{definition}

Let $k\ge 2$ be an integer. We let $A_k$ denote the alphabet $\{0,1,\ldots,k-1\}$. For the usual base-$k$ numeration system built on the language
\begin{equation}
  \label{eq:Lk}
L_k:=\{\varepsilon\}\cup \{1,\ldots,k-1\}\{0,\ldots,k-1\}^*,   
\end{equation}
an $\mathcal{S}$-automatic sequence is said to be $k$-automatic \cite{Allouche--Shallit-2003}.
We also write $\rep_k$ and $\val_k$ in this context.

\subsection{Tree languages}\label{ss:13}
Prefix-closed languages define labeled trees (also called {\em trie} or {\em prefix-tree} in computer science) and vice-versa. Let ${(A,<)}$ be a totally ordered (finite) alphabet and let $L$ be a prefix-closed language over $(A,<)$. The set of nodes of the tree is $L$. If $w$ and $wd$ are words in $L$ with $d\in A$, then there is an edge from $w$ to $wd$ with label~$d$. The children of a node are ordered by the labels of the letters in the ordered alphabet $A$. In Figure~\ref{fig:ab}, we have depicted the first levels of the tree associated with the prefix-closed language $a^*b^*$. Nodes are enumerated by breadth-first traversal (or, serialization).
  \begin{figure}[h!tb]
    \centering
  \tikzset{
  s_bla/.style = {circle,fill=white,thick, draw=black, inner sep=0pt, minimum size=12pt},
  s_red/.style = {circle,black,fill=gray,thick, inner sep=0pt, minimum size=5pt}
}
\begin{tikzpicture}[->,>=stealth',level/.style={sibling distance = 4cm/#1},level distance = 1cm]
  \node [s_bla] {0}
  child {node [s_bla] {1}
    child {node [s_bla] {3}
      child {node [s_bla] {6} edge from parent node[left] {$a$}}
      child {node [s_bla] {7} edge from parent node[right] {$b$}}
      edge from parent node[above] {$a$}}
    child {node [s_bla] {4}
      child {node [s_bla] {8} edge from parent node[right] {$b$}}
      edge from parent node[above] {$b$}}
    edge from parent node[above] {$a$}} 
  child {node [s_bla] {2}
    child {node [s_bla] {5}
      child {node [s_bla] {9} edge from parent node[right] {$b$}}
      edge from parent node[right] {$b$}}
    edge from parent node[above] {$b$}}
  ;
\end{tikzpicture}    
    \caption{The first few levels of the tree associated with $a^*b^*$.}
    \label{fig:ab}
  \end{figure}
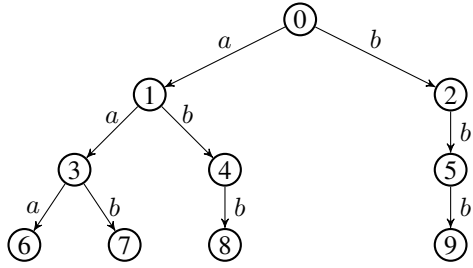

We recall some notion from \cite{Marsault--Sakarovitch-2} or \cite{Marsault--Sakarovitch-2017}. Let $T$ be an ordered tree of finite degree. The {\em (breadth-first) signature} of $T$ is a sequence of integers, the sequence of the degrees of the nodes visited by the (canonical) breadth-first traversal of the tree. The {\em (breadth-first) labeling} of $T$ is the infinite sequence of the labels of the edges visited by the breadth-first traversal of this tree. As an example, with the tree in Figure~\ref{fig:ab}, its signature is $2,2,1,2,1,1,2,1,1,1,2,\ldots$ and its labeling is $a,b,a,b,b,a,b,b,b,a,b,\ldots$.

  \begin{remark}\label{rem:itree}
  As observed by Marsault and Sakarovitch~\cite{Marsault--Sakarovitch-2}, it is usually convenient to consider {\em i-trees}: the root is assumed to be a child of itself.  It is especially the case for positional numeration systems when one has to deal with leading zeroes as the words $u$ and $0u$ may represent the same integer.
In an i-tree, paths labeled by $u$ and $0u$ lead to the same node.
  \end{remark}

    We now present a useful way to describe or generate infinite labeled i-trees.
  Let $A$ be a finite alphabet of which $0$ is assumed to be the smallest letter.
  A {\em  labeled signature} is an infinite sequence $(w_n)_{n\ge 0}$ of finite words over $A$ providing a signature $(|w_n|)_{n\ge 0}$ and a consistent labeling of a tree (made of the sequence of letters of $(w_n)_{n\ge 0}$). It will be assumed that the letters of each word are in strictly increasing order and that $w_0=0x$ with $x\in A^+$.
To that aim we let $\mathsf{inc}(A^*)$ denote the set of words over $A$ with increasingly ordered letters.
  For instance, $025$ belongs to $\mathsf{inc}(A_6^*)$ but $0241$ does not.
  Examples of labeled signatures will be given in Section~\ref{subsec:rational bases}.

\begin{remark}\label{rk: enumeration sommets arbre} 
Since a labeled signature $\mathsf{s}$ generates an i-tree, by abuse, we say that such a signature defines a prefix-closed language denoted by $L(\mathsf{s})$, which is made of the labels, not starting with $0$, of the paths in the i-tree.
Moreover, since we assumed the words of $\mathsf{s}$ all belong to $\mathsf{inc}(A^*)$ for some finite alphabet $A$, the canonical breadth-first traversal of this tree produces an abstract numeration system. Indeed the enumeration of the nodes $v_0,v_1,v_2,\ldots$ of the tree is such that $v_n$ is the $n$th word in the radix ordered language $L(\mathsf{s})$. The language $L(\mathsf{s})$, the set of nodes of the tree and $\mathbb{N}$ are thus in one-to-one correspondence.
    \end{remark}

\subsection{Rational bases}\label{subsec:rational bases}
The framework of rational base numeration systems \cite{Akiyama--Frougny-Sakarovitch-2008} is an interesting setting giving rise to a non-regular numeration language. Nevertheless the corresponding tree has a rich combinatorial structure: it has a purely periodic labeled signature.

Let $p$ and $q$ be two relatively prime integers with $p > q > 1$. 
Given a positive integer $n$, we define the sequence $(n_i)_{i\ge 0}$ as follows: we set $n_0 = n$ and, for all $i\ge 0$,
$qn_i=pn_{i+1}+a_i$ where $a_i$ is the remainder of the Euclidean division of $qn_i$ by $p$.
Note that $a_i\in A_p$ for all $i\ge 0$.
Since $p > q$, the sequence $(n_i)_{i\ge 0}$ is decreasing and eventually vanishes at some index $\ell+1$.
We obtain
\[
n = \sum_{i=0}^\ell \frac{a_i}{q}\left(\frac{p}{q}\right)^i.
\]
Conversely, for a word $w=w_\ell w_{\ell-1} \cdots w_0\in A_p^*$, the value of $w$ in base $\frac{p}{q}$ is the rational number
\[
\val_{\frac{p}{q}}(w)=\sum_{i=0}^\ell \frac{w_i}{q}\left(\frac{p}{q}\right)^i.
\]
Note that $\val_{\frac{p}{q}}(w)$ is a not always an integer and $\val_{\frac{p}{q}}(uv)=\val_{\frac{p}{q}}(u)(\frac{p}{q})^{|v|}+\val_{\frac{p}{q}}(v)$ for all $u,v\in A_p^*$.
We let $N_{\frac{p}{q}}$ denote the {\em value set}, i.e., the set of numbers representable in base $\frac{p}{q}$:
$$N_{\frac{p}{q}}=\val_{\frac{p}{q}}(A_p^*)=\left\{x\in\mathbb{Q}\mid \exists w\in A_p^* : \val_{\frac{p}{q}}(w)=x\right\}.$$

A word $w\in A_p^*$ is a {\em representation} of an integer $n\ge 0$ in base~$\frac{p}{q}$ if $\val_{\frac{p}{q}}(w)=n$.
Just as for integer bases, representations in rational bases are unique up to leading zeroes~\cite[Theorem~1]{Akiyama--Frougny-Sakarovitch-2008}.
Therefore we let $\rep_{\frac{p}{q}}(n)$ denote the representation of $n$ in base $\frac{p}{q}$ that does not start with $0$.
By convention, the representation of $0$ in base $\frac{p}{q}$ is the empty word $\varepsilon$.
In base~$\frac{p}{q}$, the numeration language is the set $$L_{\frac{p}{q}}=\left\{\rep_{\frac{p}{q}}(n)\mid n\ge 0\right\}.$$ 
Hence, rational base numeration systems are special cases of ANS built on $L_{\frac{p}{q}}$: $m<n$ if and only if $\rep_{\frac{p}{a}}(m)<\rep_{\frac{p}{a}}(n)$ for the radix order. 
It is clear that $L_{\frac{p}{q}}\subseteq A_p^*$ is a prefix-closed language. As a consequence of the previous section, it can be seen as a tree.

 \begin{example} The alphabet for the base $\frac{3}{2}$ is $A_3=\{0,1,2\}$.
The first few words in $L_{\frac{3}{2}}$ are
$\varepsilon$,  $2$,  $21$,  $210$,  $212$,  $2101$,  $2120$,  $2122$, and the associated i-tree is depicted in Figure~\ref{fig:tree32}. 
If we add an edge of label $0$ on the root of this tree (see Remark~\ref{rem:itree}), its signature is $2,1,2,1,\ldots$ and its  labeling is $0,2,1,0,2,1,0,2,1,\ldots$. Otherwise stated, the purely periodic labeled signature $(02,1)^\omega$ gives the i-tree of the language $L_{\frac{3}{2}}$; see Figure~\ref{fig:tree32}. For all $n\ge 0$, the $n$th node in the breadth-first traversal is the word $\rep_{\frac32}(n)$.
Observe that there is an edge labeled by $a\in A_3$ from the node $n$ to the node $m$ if and only if $m=\frac{3}{2} \cdot n + \frac{a}{2}$. This remark is valid for all rational bases.
\begin{figure}[h!t]
  \centering
  \tikzset{
  s_bla/.style = {circle,fill=white,thick, draw=black, inner sep=0pt, minimum size=12pt},
  s_red/.style = {circle,black,fill=gray,thick, inner sep=0pt, minimum size=5pt}
}
\begin{tikzpicture}[->,>=stealth',level/.style={sibling distance = 6cm/#1},level distance = 1cm]
  \node (root) [s_bla] {0}
  child {node [s_bla] {1} 
  child {node [s_bla] {2} 
      child {node [s_bla] {3} 
        child {node [s_bla] {5} 
          child {node [s_bla] {8} edge from parent node[left] {$1$} }
          edge from parent node[left] {$1$} }
      edge from parent node[left] {$0$} }
    child {node [s_bla] {4} 
      child {node [s_bla] {6}
        child {node [s_bla] {9} edge from parent node[left] {$0$}}
        child {node [s_bla] {10} edge from parent node[right] {$2$}}
        edge from parent node[left] {$0$}} 
      child {node [s_bla] {7}
        child {node [s_bla] {11} edge from parent node[right] {$1$} }
        edge from parent node[right] {$2$}} 
      edge from parent node[right] {$2$}}
    edge from parent node[right] {$1$}}
   edge from parent node[right] {$2$}}
;

  \draw (root) to [loop right] node [right] {$0$} (root);
\end{tikzpicture}
  \caption{The first levels of the i-tree associated with $L_\frac32$.}
  \label{fig:tree32}
\end{figure}
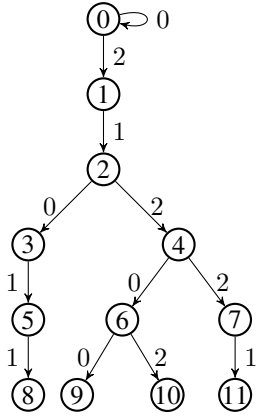
\end{example}

\begin{remark}\label{rem:perT}
The language $L_{\frac{p}{q}}$ is highly non-regular: it has the bounded left-iteration property; for details, see \cite{Marsault--Sakarovitch-1}. 
In $L_{\frac{p}{q}}$ seen as a tree, no two infinite subtrees are isomorphic, i.e., for any two words $u,v\in L_{\frac{p}{q}}$ with $u\neq v$, the quotients $u^{-1}L_{\frac{p}{q}}$ and $v^{-1}L_{\frac{p}{q}}$ are distinct. 
As we will see with Lemma~\ref{lem:414}, this does not prevent the languages $u^{-1}L_{\frac{p}{q}}$ and $v^{-1}L_{\frac{p}{q}}$ from coinciding on words of length bounded by a constant depending on $\val_{\frac{p}{q}}(u)$ and $\val_{\frac{p}{q}}(v)$ modulo a power of $q$.
Nevertheless the associated tree has a purely periodic labeled signature.  For example, with $\frac{p}{q}$ respectively equal to $\frac{3}{2}$, $\frac{5}{2}$, $\frac{7}{3}$ and $\frac{11}{4}$, we respectively have the signatures $(02,1)^\omega$, $(024,13)^\omega$, $(036,25,14)^\omega$, $(048,159,26(10),37)^\omega$. Generalizations of these languages (called rhythmic generations of trees) are studied in \cite{Marsault--Sakarovitch-2017}.
\end{remark}

\begin{definition}
 We say that a sequence is {\em $\frac{p}{q}$-automatic} if it is $\mathcal{S}$-automatic for the ANS built on the language $L_\frac{p}{q}$, i.e., $\mathcal{S}=(L_\frac{p}{q},A_p,<)$. 
\end{definition}

\subsection{Alternating morphisms}

The Kolakoski--Oldenburger word \cite[\texttt{A000002}]{OEIS} is the unique word $\mathbf{k}$ over $\{1,2\}$ starting with $2$ and satisfying $\Delta(\mathbf{k})=\mathbf{k}$ where $\Delta$ is the run-length encoding map $$\mathbf{k}=2211212212211\cdots.$$ It is a well-known (and challenging) object of study in combinatorics on words. It can be obtained by periodically iterating two morphisms, namely
$$h_0:\left\{\begin{array}{l}
               1\mapsto 2\\
               2\mapsto 22\\
             \end{array}\right.\quad\text{ and }\quad
h_1: \left\{\begin{array}{l}
               1\mapsto 1\\
               2\mapsto 11.\\
             \end{array} \right.$$          
More precisely, in \cite{kola}, 
$\mathbf{k}=k_0k_1k_2\cdots$ is expressed as the fixed point of the iterated morphisms $(h_0,h_1)$, i.e., $$\mathbf{k}=h_0(k_0)h_1(k_1)\cdots h_0(k_{2n}) h_1(k_{2n+1})\cdots.$$
In the literature, one also finds the terminology PD0L for {\em D0L system with periodic control} \cite{Endrullis,Lepisto}.

\begin{definition}
Let $r\ge 1$ be an integer, let $A$ be a finite alphabet, and let $f_0,\ldots,f_{r-1}$ be $r$ morphisms over $A^*$. An infinite word $\mathbf{w}=w_0w_1w_2\cdots$ over $A$ is an {\em alternating fixed point} of $(f_0,\ldots,f_{r-1})$ if 
  $$\mathbf{w}=f_0(w_0)f_1(w_1)\cdots f_{r-1}(w_{r-1})f_0(w_r)\cdots f_{i\bmod{r}}(w_i)\cdots .$$
\end{definition}

As observed by Dekking \cite{Dek} for the Kolakoski word, an alternating fixed point can also be obtained by an $r$-block substitution.

\begin{definition}
  Let $r\ge 1$ be an integer and let $A$ be a finite alphabet. An {\em$r$-block substitution} $g:A^r\to A^*$ maps a word $w_0\cdots w_{rn-1} \in A^*$ to 
  \[
  g(w_0\cdots w_{r-1}) g(w_r\cdots w_{2r-1})\cdots g(w_{r(n-1)}\cdots w_{rn-1}).
  \] 
  If the length of the word is not a multiple of $r$, then the suffix of the word is ignored under the action of $g$. An infinite word $\mathbf{w}=w_0w_1w_2\cdots$ over $A$ is a {\em fixed point of the $r$-block substitution} $g:A^r\to A^*$ if
  $$\mathbf{w}= g(w_0\cdots w_{r-1}) g(w_r\cdots w_{2r-1})\cdots.$$
\end{definition}

\begin{prop}\label{pro:rblock}
Let $r\ge 1$ be an integer, let $A$ be a finite alphabet, and let $f_0,\ldots,f_{r-1}$ be $r$ morphisms over $A^*$. If an infinite word over $A$ is an alternating fixed point of $(f_0,\ldots,f_{r-1})$, then it is a fixed point of an $r$-block substitution.
\end{prop}

\begin{proof}
For every of length-$r$ word $a_0\cdots a_{r-1}\in A^*$, define the $r$-block substitution $g:A^r\to A^*$ by
$g(a_0\cdots a_{r-1})=f_0(a_0)\cdots f_{r-1}(a_{r-1})$.
\end{proof}

Thanks to the previous result, the Kolakoski--Oldenburger word $\mathbf{k}$ is also a fixed point of the $2$-block substitution
$$g:\left\{
  \begin{array}{l}
    11\mapsto h_0(1)h_1(1)=21\\
    12\mapsto h_0(1)h_1(2)=211\\
    21\mapsto h_0(2)h_1(1)=221\\
    22\mapsto h_0(2)h_1(2)=2211.\\
  \end{array}\right.$$
  Observe that the lengths of images under $g$ are not all equal.

\section{Concrete examples of automatic sequences}\label{sec:ASRB}

Let us present how the above concepts are linked with the help of some examples.
The first one is our toy example.

\begin{example}\label{exa:tm32}
Let $(s(n))_{n\ge 0}$ be the sum-of-digits in base $\frac{3}{2}$. This sequence was, in particular, studied in \cite{sod}.
We have 
\[
(s(n))_{n\ge 0}=0, 2, 3, 3, 5, 4, 5, 7, 5, 5, 7, 8, 5, 7, 6, 7, 9, \ldots.
\]
We let $\mathbf{t}$ denote the sequence $(s(n)\bmod{2})_{n\ge 0}$, 
$$\mathbf{t}=00111011111011011\cdots.$$
The sequence $\mathbf{t}$ is $\frac{3}{2}$-automatic as the DFAO in Figure~\ref{Fig: DFAO sum-of-digits mod 2} generates $\mathbf{t}$ when reading base-$\frac{3}{2}$ representations.
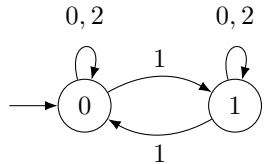
\begin{figure}[htb]
\begin{center}
\begin{tikzpicture}
\tikzstyle{every node}=[shape=circle, fill=none, draw=black,minimum size=20pt, inner sep=2pt]
\node(1) at (0,0) {$0$};
\node(2) at (2,0) {$1$};

\tikzstyle{every node}=[shape=circle, minimum size=5pt, inner sep=2pt]

\draw [-Latex] (-1,0) to node [above] {} (1);

\draw [-Latex] (1) to [loop above] node [above] {$0,2$} (1);
\draw [-Latex] (1) to [bend left] node [above] {$1$} (2);
\draw [-Latex] (2) to [loop above] node [above] {$0,2$} (2);
\draw [-Latex] (2) to [bend left] node [below] {$1$} (1);

\end{tikzpicture}
\end{center}
\caption{A DFAO generating the sum-of-digits in base $\frac{3}{2}$ modulo $2$.}
\label{Fig: DFAO sum-of-digits mod 2}
\end{figure}

As a consequence of Proposition~\ref{pro:aut-mor}, it will turn out that $\mathbf{t}$ is an alternating fixed point of $(f_0,f_1)$ with
\begin{equation}
  \label{eq:mortoy}
  f_0 \colon \left\{\begin{array}{l}
                      0 \mapsto 00\\
                      1\mapsto 11\\
                    \end{array}\right.\quad \text{ and }
  f_1\colon \left\{\begin{array}{l}
                      0 \mapsto 1\\
                      1\mapsto 0.\\
                    \end{array}\right.
\end{equation}
With Proposition~\ref{pro:rblock}, $\mathbf{t}$ is also a fixed point
of the $2$-block substitution
$$g:\left\{
  \begin{array}{l}
    00\mapsto f_0(0)f_1(0)=001\\
    01\mapsto f_0(0)f_1(1)=000\\
    10\mapsto f_0(1)f_1(0)=111\\
    11\mapsto f_0(1)f_1(1)=110.\\
  \end{array}\right.$$
Observe that we have a $2$-block substitution with images of length $3$. This is not a coincidence, as we will see with Corollary~\ref{cor:cas32}.      
\end{example}

Automatic sequences in integer bases are morphic words, i.e., images, under a coding, of a fixed point of a prolongable morphism~\cite{Allouche--Shallit-2003}. As shown by the next example, there are $\frac{3}{2}$-automatic sequences that are not morphic. 
For a word $u\in\{0,1\}^*$, we let $\overline{u}$ denote the word obtained by applying the involution $i\mapsto 1-i$, $i\in\{0,1\}$, to the letters of $u$. 

\begin{example}\label{ex: Lepisto}
Lepist\"o considered in \cite{Lepisto} the following $2$-block substitution
$$h_2:\left\{
  \begin{array}{l}
    00\mapsto g_0(0)\overline{0}=011\\
    01\mapsto g_0(0)\overline{1}=010\\
    10\mapsto g_0(1)\overline{0}=001\\
    11\mapsto g_0(1)\overline{1}=000\\
  \end{array}\right. \text{ with } g_0:0\mapsto 01, 1\mapsto 00,$$
producing the word $\mathbf{F}_2=01001100001\cdots$. He showed that the factor complexity $\mathsf{p}_{\mathbf{F}_2}$ of this word satisfies $\mathsf{p}_{\mathbf{F}_2}(n)> \delta n^t$ for some $\delta>0$ and $t>2$. Hence, this word cannot be purely morphic nor morphic (because these kinds of words have a factor complexity in $O(n^2)$ \cite{Pansiot}). With Proposition~\ref{pro:mor-aut}, we can show that $\mathbf{F}_2$ is a $\frac{3}{2}$-automatic sequence generated by the DFAO depicted in Figure~\ref{Fig:CE}.
\begin{figure}[htb]
\begin{center}
\begin{tikzpicture}
\tikzstyle{every node}=[shape=circle, fill=none, draw=black,minimum size=20pt, inner sep=2pt]
\node(1) at (0,0) {$0$};
\node(2) at (2,0) {$1$};

\tikzstyle{every node}=[shape=circle, minimum size=5pt, inner sep=2pt]

\draw [-Latex] (-1,0) to node [above] {} (1);

\draw [-Latex] (1) to [loop above] node [left=0.1] {$0$} (1);
\draw [-Latex] (1) to [bend left] node [above] {$1,2$} (2);
\draw [-Latex] (2) to [bend left] node [below] {$0,1,2$} (1);

\end{tikzpicture}
\end{center}
\caption{A DFAO generating $\mathbf{F}_2$.}
\label{Fig:CE}
\end{figure}
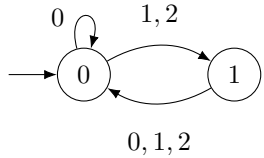
\end{example}

\begin{remark}
  Similarly, the non-morphic word $\mathbf{F}_p$ introduced in \cite{Lepisto} is $\frac{p+1}{p}$-automatic. It is generated by the $p$-block substitution defined by $h_p(au)=g_0(a)\overline{u}$ for $a\in\{0,1\}$ and $u\in\{0,1\}^{p-1}$, where $g_0$ is defined in Example~\ref{ex: Lepisto}.
\end{remark}

We conclude this section with an example of an automatic sequence associated with a language coming from a periodic signature.

\begin{example}
Consider the periodic labeled signature $\mathsf{s}=(023,14,5)^\omega$ producing the i-tree in Figure~\ref{fig:tree023-14-5}.
The first few words in $L(\mathsf{s})$ are $\varepsilon$, $2$,  $3$,  $21$,  $24$,  $35$,  $210$,  $212$,  $213$,  $241$,  $244$,  $355$,  which give the representations of the first $12$ integers in the abstract numeration system $\mathcal{S}=(L(\mathsf{s}),A_6,<)$.
For instance, $\rep_\mathcal{S}(15)=2121$ as the path of label $2121$ leads to the node $15$ in Figure~\ref{fig:tree023-14-5}.
The sum-of-digits in $\mathcal{S}$ modulo $2$, starting with
\[
001100110101\cdots,
\]
is $\mathcal{S}$-automatic since it is generated by the DFAO in Figure~\ref{Fig: DFAO sum-of-digits mod 2 in S}.
As a consequence of Proposition~\ref{pro:aut-mor} and Theorem~\ref{thm:cobham}, we will see that this sequence is also the coding of an alternating fixed point of three morphisms.

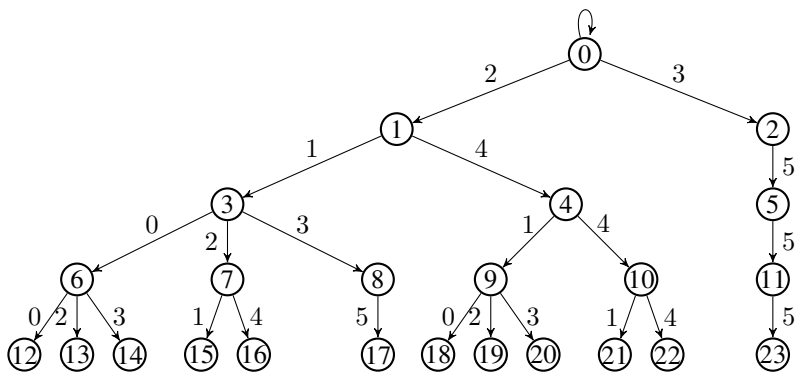
\begin{figure}[h!t]
  \centering
  \tikzset{
  s_bla/.style = {circle,fill=white,thick, draw=black, inner sep=0pt, minimum size=12pt},
  s_red/.style = {circle,black,fill=gray,thick, inner sep=0pt, minimum size=5pt}
}
\tikzstyle{level 1}=[sibling distance=50mm]
\tikzstyle{level 2}=[sibling distance=45mm]
\tikzstyle{level 3}=[sibling distance=20mm]
\tikzstyle{level 4}=[sibling distance=7mm]
\begin{tikzpicture}[->,>=stealth',level distance = 1cm]
  \node (root) [s_bla] {0}
  child {node [s_bla] {1} 
  child {node [s_bla] {3} 
      child {node [s_bla] {6} 
        child {node [s_bla] {12} 
          edge from parent node[left] {$0$} }
          child {node [s_bla] {13} 
          edge from parent node[left] {$2$} }
			child {node [s_bla] {14} 
          edge from parent node[right] {$3$} }
      edge from parent node[above] {$0$} }
            child {node [s_bla] {7} 
        child {node [s_bla] {15} 
          edge from parent node[left] {$1$} }
          child {node [s_bla] {16} 
          edge from parent node[right] {$4$} }
      edge from parent node[left] {$2$} }
       child {node [s_bla] {8} 
        child {node [s_bla] {17} 
          edge from parent node[left] {$5$} }
      edge from parent node[above] {$3$} }
    edge from parent node[above] {$1$}}
  child {node [s_bla] {4} 
      child {node [s_bla] {9} 
        child {node [s_bla] {18} 
          edge from parent node[left] {$0$} }
          child {node [s_bla] {19} 
          edge from parent node[left] {$2$} }
          child {node [s_bla] {20} 
          edge from parent node[right] {$3$} }
      edge from parent node[above] {$1$} }
    child {node [s_bla] {10} 
      child {node [s_bla] {21} 
          edge from parent node[left] {$1$} }
          child {node [s_bla] {22} 
          edge from parent node[right] {$4$} }
      edge from parent node[above] {$4$}}
    edge from parent node[above] {$4$}}
   edge from parent node[above] {$2$}}
    child {node [s_bla] {2} 
  child {node [s_bla] {5} 
      child {node [s_bla] {11} 
          child {node [s_bla] {23} 
          edge from parent node[right] {$5$} }
      edge from parent node[right] {$5$} }
    edge from parent node[right] {$5$}}
   edge from parent node[above] {$3$}}
;

  \path (root) edge [loop above] (root);
\end{tikzpicture}
  \caption{The i-tree associated with the signature $(023,14,5)^\omega$.}
  \label{fig:tree023-14-5}
\end{figure}  

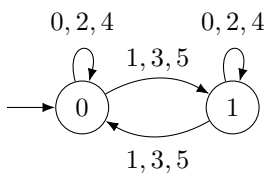
\begin{figure}[htb]
\begin{center}
\begin{tikzpicture}
\tikzstyle{every node}=[shape=circle, fill=none, draw=black,minimum size=20pt, inner sep=2pt]
\node(1) at (0,0) {$0$};
\node(2) at (2,0) {$1$};

\tikzstyle{every node}=[shape=circle, minimum size=5pt, inner sep=2pt]

\draw [-Latex] (-1,0) to node [above] {} (1);

\draw [-Latex] (1) to [loop above] node [above=-0.2] {$0,2,4$} (1);
\draw [-Latex] (1) to [bend left] node [above=-0.3] {$1,3,5$} (2);
\draw [-Latex] (2) to [loop above] node [above=-0.2] {$0,2,4$} (2);
\draw [-Latex] (2) to [bend left] node [below=-0.2] {$1,3,5$} (1);

\end{tikzpicture}
\end{center}
\caption{A DFAO generating the sum-of-digits modulo $2$ in the ANS $\mathcal{S}=(L(\mathsf{s}),A_6,<)$ where $\mathsf{s}=(023,14,5)^\omega$.}
\label{Fig: DFAO sum-of-digits mod 2 in S}
\end{figure}
  
\end{example}

\section{Cobham's theorem}\label{sec:cobham}

Cobham's theorem from 1972 states that a sequence is $k$-automatic if and only if it is the image under a coding of the fixed point of a $k$-uniform morphism \cite{Cobham73} (or see \cite[Theorem~6.3.2]{Allouche--Shallit-2003}). This result has been generalized to various contexts: numeration systems associated with a substitution, Pisot numeration systems, Bertrand numeration systems, ANS with regular languages, and so on \cite{BH,DT,Massuir,Maes}. Also see \cite{LR} or \cite{Rigo} for a comprehensive presentation. In this section, we adapt it to the case of $\mathcal{S}$-automatic sequences built on tree languages with a periodic labeled signature (so, in particular, to the rational base case). We start off with a technical lemma.

\begin{lemma}\label{lem:iteration}
Let $r\ge 1$ be an integer, let $A$ be a finite alphabet, and let $f_0,\ldots,f_{r-1}$ be morphisms over $A^*$. Let $\mathbf{x}=x_0x_1x_2\cdots$ be an alternating fixed point of $(f_0,\ldots,f_{r-1})$.
For all $m\ge 0$, we have 
$$f_{m \bmod{r}}(x_m)= x_i \cdots x_{i+|f_{m \bmod{r}}(x_m)|-1}$$
  where $i=\sum_{j=0}^{m-1} \left| f_{j \bmod{r}}(x_j) \right|$.
\end{lemma}

\begin{proof}
Let $m\ge 0$.
From the definition of an alternating fixed point, we have the factorization $\mathbf{x}=u f_{m \bmod{r}}(x_m) f_{(m+1) \bmod{r}}(x_{m+1})\cdots$ where 
\[
u=f_0(x_0) f_1(x_1)\cdots f_{r-1}(x_{r-1})f_0(x_r)\cdots f_{(m-1) \bmod{r}}(x_{m-1}).
\]
Now $|u|=\sum_{j=0}^{m-1} \left| f_{j \bmod{r}}(x_j) \right|$, which concludes the proof.
\end{proof}

Given an $\mathcal{S}$-automatic sequence associated with the language of a tree with a purely periodic labeled signature, we can turn it into an alternating fixed point of uniform morphisms.

\begin{prop}\label{pro:aut-mor}
Let $r\ge 1$ be an integer and let $A$ be a finite alphabet of digits. 
Let $w_0,\ldots,w_{r-1}$ be $r$ non-empty words in $\mathsf{inc}(A^*)$.
Consider the language $L(\mathsf{s})$ of the i-tree generated by the purely periodic signature $\mathsf{s}=(w_0,w_1,\ldots,w_{r-1})^\omega$.
Let $\mathcal{A}=(Q,q_0,A,\delta)$ be a DFA. For $i\in\{0,\ldots,r-1\}$, we define the $r$ morphisms from $Q^*$ to itself by
$$f_i:Q\to Q^{|w_i|}, q\mapsto \delta(q,w_{i,0})\cdots \delta(q,w_{i,|w_i|-1}),$$
where $w_{i,j}$ denotes the $j$th letter of $w_i$.
The alternating fixed point $\mathbf{x}=x_0x_1\cdots$ of $(f_0,\ldots,f_{r-1})$ starting with $q_0$ is the sequence of states reached in $\mathcal{A}$ when reading the words of $L(\mathsf{s})$ in increasing radix order, i.e., for all $n\ge 0$, $x_n=\delta(q_0,\rep_\mathcal{S}(n))$ with $\mathcal{S}=(L(\mathsf{s}),A,<)$.
\end{prop}

\begin{proof} 
Up to renaming the letters of $w_0$, without loss of generality we may assume that $w_0=0x$ with $x\in A^+$.

We proceed by induction on $n\ge 0$. It is clear that $x_0=\delta(q_0,\varepsilon)=q_0$.
Let $n\ge 1$.
  Assume that the property holds for all integers less than $n$ and we prove it for $n$. 

  Write $\rep_\mathcal{S}(n)=a_\ell \cdots a_1 a_0$. This means that in the i-tree generated by $\mathsf{s}$, we have a path of label $a_\ell\cdots a_0$ from the root. We identify words in $L(\mathsf{s})$ with vertices of the i-tree.

  Since $L(\mathsf{s})$ is prefix-closed, there exists an integer $m<n$ such that $\rep_\mathcal{S}(m)=a_\ell \cdots a_1$.
   Let $i=m\bmod r$. By definition of the periodic labeled signature $\mathsf{s}$, in the i-tree generated by $\mathsf{s}$,  reading $a_\ell\cdots a_1$ from the root leads to a node having $|w_i|$ children that are reached with edges labeled by the letters of $w_i$. 
Since $w_i \in \mathsf{inc}(A^*)$, the letter $a_0$ occurs exactly once in $w_i$, so assume that $w_{i,j}=a_0$ for some $j\in \{0,\ldots,|w_i|-1\}$. By construction of the i-tree given by a periodic labeled signature (see Figure~\ref{fig:eq31} for a pictorial description), we have that 
  \begin{equation}
    \label{eq:longarbre}
    n=\sum_{\substack{v\in L(\mathsf{s})\\ v<\rep_\mathcal{S}(m)}} \deg(v)+j=\sum_{k=0}^{m-1} |w_{k \bmod{r}}|+j.
  \end{equation}
By the induction hypothesis, we obtain
  $$\delta(q_0,\rep_\mathcal{S}(n))=\delta(\delta(q_0,\rep_\mathcal{S}(m)),a_0)=\delta(x_m,a_0)$$
and by definition of $f_i$, we get $\delta(x_m,a_0)=[f_i(x_m)]_j=[f_{m\bmod{r}}(x_m)]_j$.
From Lemma~\ref{lem:iteration} and Equation \eqref{eq:longarbre}, this is exactly $x_n$, as desired.
\end{proof}

\begin{figure}[h!]
  \centering
  \tikzset{
    s_bla/.style = {circle, fill=black, draw, thick, inner sep=1pt, minimum size=3pt},
      s_red/.style = {circle, fill=red, draw, thick, inner sep=1pt, minimum size=4pt},
}
\begin{tikzpicture}
  \node at (0,0) [s_bla, label=above:$0$] (a1) {};
  \node at (-3,-2) [s_bla, label=above:$1$] (a2) {};
  \node at (-.5,-2) [s_red] (b1) {};
  \node at (.5,-4) [s_red] (b2) {};
  \node at (1,-6) [s_red,label=above:$m$] (b3) {};
  \node at (2.5,-6) [s_bla,label=above:$m+1$] (c1) {};
   \node at (4.5,-6) [s_bla] (c2) {};
  \node at (.8,-8) [s_red,label=right:$n$] (b4) {};
  \node at (.3,-8) [s_bla] (b5) {};
  \node at (1.9,-8) [s_bla] (b6) {};
  \node at (-1.5,-6) [s_bla,label=above:$m-1$] (b3') {};
  \node at (-2.8,-8) [s_bla] (b5') {};
  \node at (-.9,-8) [s_bla] (b6') {};
  \node at (-5.5,-6) [s_bla] (b3'') {};
  \node at (-6.8,-8) [s_bla] (b5'') {};
  \node at (-4.9,-8) [s_bla] (b6'') {};
  \node at (1.2,-2) [s_bla, label=above right:$|w_0|-1$] (a3) {};
  \node at (-4.5,-4) [s_bla] (a4) {};
  \node at (-2,-4) [s_bla] (a5) {};
  \node at (2,-4) [s_bla] (a6) {};
  \node at (3.5,-4) [s_bla] (a7) {};
  \draw [dashed] (a2) -- (a3);
  \draw [dashed] (a4) -- (a7);
  \draw [dashed] (b3'') -- (c2);
  \draw [dashed] (b5'') -- (b4);
  \draw [dotted] (a4) -- (b3'');
  \draw [dotted] (a7) -- (c2);
    \draw [dotted,draw=red,thick] (b2) -- (b3);
  \draw [-Latex,draw=red,thick] (a1) to [] node [right] {$a_\ell$} (b1);
  \draw [-Latex,draw=red,thick] (b1) to [] node [right] {$a_{\ell-1}$} (b2);
  \draw [-Latex,draw=red,thick] (b3) to [] node [right] {$a_0$} (b4);
  \draw [-Latex] (b3) to [] node [right] {} (b5);
  \draw [-Latex] (b3) to [] node [right] {} (b6);
  \draw [-Latex] (b3') to [] node [right] {} (b5');
  \draw [-Latex] (b3') to [] node [right] {} (b6');
  \draw [-Latex] (b3'') to [] node [right] {} (b5'');
  \draw [-Latex] (b3'') to [] node [right] {} (b6'');
  \draw [-Latex] (a1) to [] node [right] {} (a2);
  \draw [-Latex] (a1) to [] node [right] {} (a3);
  \draw [-Latex] (a2) to [] node [right] {} (a4);
  \draw [-Latex] (a2) to [] node [right] {} (a5);
    \draw [-Latex] (a3) to [] node [right] {} (a6);
    \draw [-Latex] (a3) to [] node [right] {} (a7);
    \draw [-Latex] (a1) to [loop right] (a1);
    \draw [blue,decorate,decoration={brace,amplitude=10pt,mirror},xshift=0.4pt,yshift=-0.4pt](-3,-2) -- (1.2,-2) node[black,midway,yshift=-0.6cm] {\footnotesize $\deg(v_0)-1$};
        \draw [blue,decorate,decoration={brace,amplitude=10pt,mirror},xshift=0.4pt,yshift=-0.4pt](-4.5,-4) -- (-2,-4) node[black,midway,yshift=-0.6cm] {\footnotesize $\deg(v_1)$};
           \draw [blue,decorate,decoration={brace,amplitude=10pt,mirror},xshift=0.4pt,yshift=-0.4pt](-1,-4) -- (1,-4) node[black,midway,yshift=-0.6cm] {}; 
           \draw [blue,decorate,decoration={brace,amplitude=10pt,mirror},xshift=0.4pt,yshift=-0.4pt](2,-4) -- (3.5,-4) node[black,midway,yshift=-0.6cm] {};
           \draw [blue,decorate,decoration={brace,amplitude=10pt,mirror},xshift=0.4pt,yshift=-0.4pt](-6.8,-8) -- (-4.9,-8) node[black,midway,yshift=-0.6cm] {};
           
           \draw [blue,decorate,decoration={brace,amplitude=10pt,mirror},xshift=0.4pt,yshift=-0.4pt](-2.8,-8) -- (-.9,-8) node[black,midway,yshift=-0.6cm] {\footnotesize $\deg(v_{m-1})$};
                \draw [blue,decorate,decoration={brace,amplitude=10pt,mirror},xshift=0.4pt,yshift=-0.4pt](.3,-8) -- (.8,-8) node[black,midway,yshift=-0.6cm] {\footnotesize $j$};
\end{tikzpicture}
  \caption{Illustration of Equation~\eqref{eq:longarbre}.}
  \label{fig:eq31}
\end{figure}
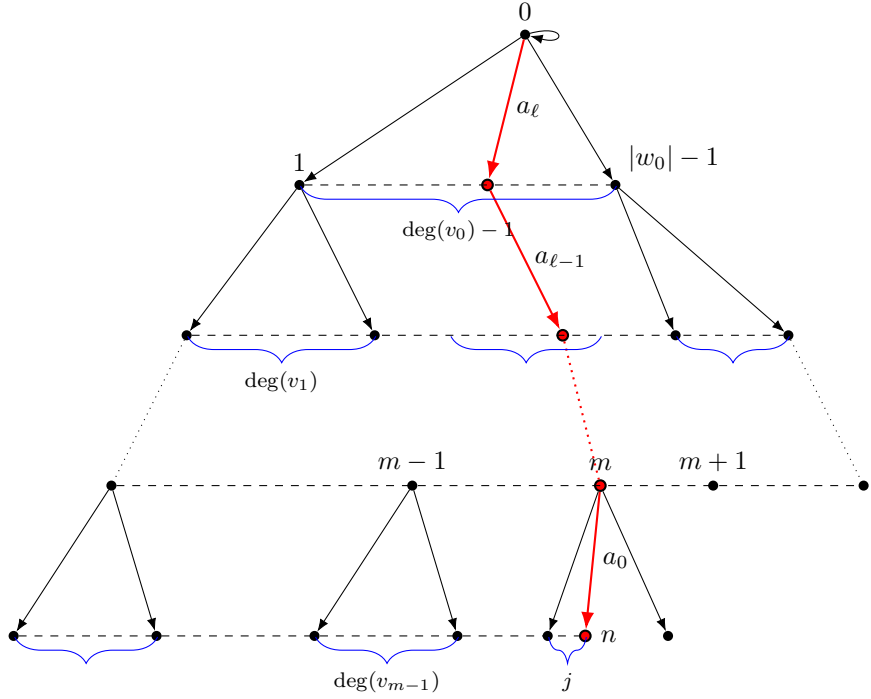

Given an alternating fixed point of uniform morphisms, we can turn it into an $\mathcal{S}$-automatic sequence for convenient choices of a language of a tree with a purely periodic labeled signature and a DFAO.

\begin{prop}\label{pro:mor-aut}
  Let $r\ge 1$ be an integer and let $A$ be a finite alphabet.
  Let $f_0,\ldots,f_{r-1}:A^*\to A^*$ be $r$ uniform morphisms of respective length $\ell_0,\ldots,\ell_{r-1}$ such that $f_0$ is prolongable on some letter $a\in A$, i.e., $f_0(a)=ax$ with $x\in A^+$. Let $\mathbf{x}=x_0x_1\cdots$ be the alternating fixed point of $(f_0,\ldots,f_{r-1})$ starting with $a$. Consider the language $L(\mathsf{s})$ of the i-tree generated by the purely periodic labeled signature
  $$\mathsf{s}=\left(0\cdots (\ell_0-1),\ell_0 (\ell_0+1)\cdots (\ell_0+\ell_1-1), \ldots, \left(\sum_{j<r-1} \ell_j \right)\cdots \left(\sum_{j<r} \ell_j-1\right)\right)^\omega,$$ which is made of consecutive non-negative integers.
  Define a DFA~$\mathcal{A}$ having
  \begin{itemize}
  \item $A$ as set of states, 
  \item $a$ as initial state,
  \item $B=\{0,\ldots,\sum_{j<r} \ell_j-1\}$ as alphabet,
  \item its transition function $\delta: A\times B \to A$ defined as follows: For all $i\in B$, there exist a unique $j_i\ge 0$ and a unique $t_i\ge 0$ such that $i=\sum_{k\le j_i-1} \ell_k +t_i$ with $t_i<\ell_{j_i}$, and we set
    $$\delta(b,i)=[f_{j_i}(b)]_{t_i},\quad \forall b\in A.$$
      \end{itemize}
  Then the word $\mathbf{x}$ is the sequence of the states reached in $\mathcal{A}$ when reading the words of $L(\mathsf{s})$ by increasing radix order, i.e., for all $n\ge 0$, $x_n=\delta(a,\rep_\mathcal{S}(n))$ with $\mathcal{S}=(L(\mathsf{s}),B,<)$.
\end{prop}

\begin{proof}
We again proceed by induction on $n\ge 0$.
It is clear that $x_0=a=\delta(a,\varepsilon)$.
Let $n\ge 1$.
Assume the property holds for all values less than $n$ and we prove it for $n$. 

Write $\rep_\mathcal{S}(n)=a_\ell \cdots a_1 a_0$. This means that in the i-tree with a periodic labeled signature $\mathsf{s}$, we have a path of label $a_\ell\cdots a_0$ from the root. We identify words in $L(\mathsf{s})\subseteq B^*$ with vertices of the i-tree.

  Since $L(\mathsf{s})$ is prefix-closed, there exists $m<n$ such that $\rep_\mathcal{S}(m)=a_\ell \cdots a_1$. 
  Let $j=m\bmod r$. In the i-tree generated by $\mathsf{s}$, reading $a_\ell\cdots a_1$ from the root leads to a node having $\ell_j$ children that are reached with edges labeled by $$\sum_{k\le j-1} \ell_k,\ \sum_{k\le j-1} \ell_k+1,\ \ldots,\ \sum_{k\le j} \ell_k-1.$$ 
Observe that the words in $\mathsf{s}$ belong to $ \mathsf{inc}(B^*)$. Therefore the letter $a_0$ occurs exactly once in $B$ and in particular amongst those labels, assume that $a_0=\sum_{k\le j-1} \ell_k+t$ for some $t\in \{0,\ldots,\ell_j-1\}$. By construction of the i-tree, we have that 
  \begin{equation}
    \label{eq:longarbre2}
    n=\sum_{\substack{v\in L(\mathsf{s})\\ v<\rep_\mathcal{S}(m)}} \deg(v)+t=\sum_{i=0}^{m-1} \ell_{i \bmod{r}}+t.
  \end{equation}
By the induction hypothesis, we obtain
  $$\delta(a,\rep_\mathcal{S}(n))=\delta(\delta(a,\rep_\mathcal{S}(m)),a_0)=\delta(x_m,a_0)$$
  and by definition of the transition function, $\delta(x_m,a_0)=[f_j(x_m)]_t=[f_{m\bmod{r}}(x_m)]_t$.
 From Lemma~\ref{lem:iteration} and Equation \eqref{eq:longarbre2}, this is exactly $x_n$.
\end{proof}

\begin{remark}
  What matters in the above statement is that two distinct words of the signature $\mathsf{s}$ do not share any common letter. It mainly ensures that the choice of the morphism to apply when defining $\delta$ is uniquely determined by the letter to be read. 
\end{remark}

\begin{example}
  If we consider the morphisms in \eqref{eq:mortoy}, Proposition~\ref{pro:mor-aut} provides us with the signature $\mathsf{s}=(01,2)^\omega$ instead of the signature $(02,1)^\omega$ of $L_{\frac32}$. We will produce the sequence $\mathbf{t}$ using the language $h(L_{\frac{3}{2}})$ where the coding $h$ is defined by $h(0)=0$, $h(1)=2$ and $h(2)=1$ and in the DFAO in Figure~\ref{Fig: DFAO sum-of-digits mod 2}, the same coding is applied to the labels of the transitions. What matters is the shape of the tree (i.e., the sequence of degrees of the vertices) rather than the labels themselves.
\end{example}

\begin{theorem}\label{thm:cobham} Let $A,B$ be two finite alphabets. 
  An infinite word over $B$ is the image under a coding $g:A\to B$ of an alternating fixed point of uniform morphisms (not necessarily of the same length) over $A$ if and only if it is $\mathcal{S}$-automatic for an abstract numeration system $\mathcal{S}$ built on a tree language with a purely periodic labeled signature.
\end{theorem}

\begin{proof}
The forward direction follows from Proposition~\ref{pro:mor-aut}: define a DFAO where the output function $\tau$ is obtained from the coding~$g:A\to B$ defined by $\tau(b)=g(b)$ for all $b$ in $A$.
The reverse direction directly follows from Proposition~\ref{pro:aut-mor}.
\end{proof}

We are able to say more in the special case of rational bases.
The tree language associated with the rational base $\frac{p}{q}$ has a periodic signature of the form $(w_0,\ldots,w_{q-1})^\omega$ with $\sum_{i=0}^{q-1} |w_i|=p$ and $w_i\in A_p^*$ for all $i$. See Remark~\ref{rem:perT} for examples.

\begin{corollary}\label{cor:cas32} 
If a sequence is $\frac{p}{q}$-automatic, then it is the image under a coding of a fixed point of a $q$-block substitution whose images all have length $p$.
\end{corollary}

\begin{proof}
Let $(w_0,\ldots,w_{q-1})^\omega$ denote the periodic signature in base $\frac{p}{q}$. Proposition~\ref{pro:aut-mor} provides $q$ morphisms $f_i$ that are respectively $|w_i|$-uniform. By Proposition~\ref{pro:rblock}, the alternating fixed point of $(f_0,\ldots,f_{q-1})$ is a fixed point of a $q$-block substitution $g$ such that, for any length-$q$ word $a_0\cdots a_{q-1}$, 
\[
|g(a_0\cdots a_{q-1})|=|f_0(a_0)f_1(a_1)\cdots f_{q-1}(a_{q-1})|=\sum_{i=0}^{q-1} |w_i|=p.
\qedhere
\]
\end{proof}

\section{Decorating trees and subtrees}\label{sec:tree}

As already observed in Section~\ref{ss:13}, a prefix-closed language $L$ over an ordered (finite) alphabet $(A,<)$ gives an ordered labeled tree $T(L)$ in which edges are  labeled by letters in $A$. Labels of paths from the root to nodes provide a one-to-one correspondence between nodes in $T(L)$ and words in~$L$. We now add an extra information, such as a color, on every node. This information is provided by a sequence taking finitely many values.

\begin{definition}
Let $T=(V,E)$ be a rooted ordered infinite tree, i.e., each node has a finite (ordered) sequence of children. 
As observed in Remark~\ref{rk: enumeration sommets arbre}, the canonical breadth-first traversal of $T$ gives an abstract numeration system --- an enumeration of the nodes: $v_0,v_1,v_2,\ldots$. Let $\mathbf{x}=x_0x_1\cdots$ be an infinite word over a finite alphabet $B$. A {\em decoration} of $T$ by $\mathbf{x}$ is a map from $V$ to $B$ associating with the node $v_n$ the decoration (or color) $x_n$, for all $n\ge 0$.
\end{definition}

To be consistent and to avoid confusion, we refer respectively to {\em label} and {\em decoration} the  labeling of the edges and nodes of a tree. 

\begin{example}
In Figure~\ref{fig:dec_tree} are depicted a prefix of $T(L_{\frac32})$ decorated with the sequence $\mathbf{t}$ of Example~\ref{exa:tm32} and a prefix of the tree $T(L_2)$ associated with the binary numeration system (see~\eqref{eq:Lk}) and decorated with the Thue--Morse sequence $0110100110010110\cdots$. In these trees, the symbol $0$ (respectively $1$) is denoted by a black (respectively red) decorated node.
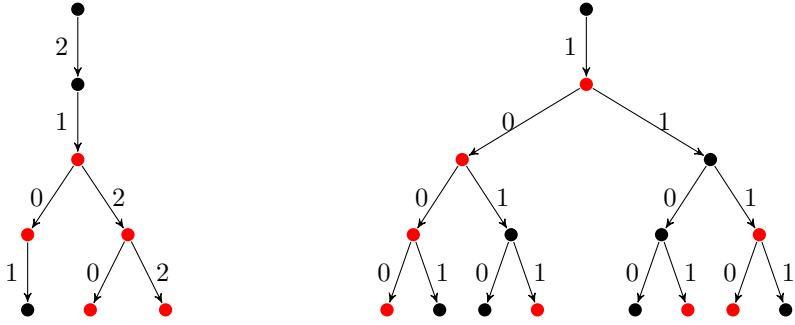
\begin{figure}[h!t]
  \centering
\begin{minipage}{.2\linewidth}
\tikzset{
  s_bla/.style = {circle,fill=black,thick, inner sep=0pt, minimum size=5pt},
  s_red/.style = {circle,black,fill=red,thick, inner sep=0pt, minimum size=5pt}
}
\begin{tikzpicture}[->,>=stealth',level/.style={sibling distance = 4cm/#1},level distance = 1cm]
  \node [s_bla] {} 
  child { node [s_bla] {} 
    child { node [s_red] {} 
      child {node [s_red] {} 
      child {node [s_bla] {} edge from parent node[left] {$1$}} 
            edge from parent node[left] {$0$}}
    child {node [s_red] {} 
      child {node [s_red] {} edge from parent node[left] {$0$}} 
      child {node [s_red] {} edge from parent node[right] {$2$}} 
  edge from parent node[right] {$2$}}
edge from parent node[left] {$1$}}
edge from parent node[left] {$2$}}  
;
\end{tikzpicture}
\end{minipage}\hskip 2cm
\begin{minipage}{.4\linewidth}
  \tikzset{
  s_bla/.style = {circle,fill=black,thick, inner sep=0pt, minimum size=5pt},
  s_red/.style = {circle,black,fill=red,thick, inner sep=0pt, minimum size=5pt}
}
\tikzstyle{level 1}=[sibling distance=50mm]
\tikzstyle{level 2}=[sibling distance=33mm]
\tikzstyle{level 3}=[sibling distance=13mm]
\tikzstyle{level 4}=[sibling distance=7mm]
\begin{tikzpicture}[->,>=stealth',level distance = 1cm]
  \node [s_bla] {} 
  child { node [s_red] {} 
    child { node [s_red] {} 
      child { node [s_red] {} 
        child { node [s_red] {} 
          edge from parent node[left] {$0$} }
        child { node [s_bla] {} 
          edge from parent node[right] {$1$} }
        edge from parent node[left] {$0$} }
      child { node [s_bla] {} 
        child { node [s_bla] {} 
          edge from parent node[left] {$0$} }
        child { node [s_red] {} 
          edge from parent node[right] {$1$} }
        edge from parent node[right] {$1$} }
      edge from parent node[left] {$0$} } 
    child { node [s_bla] {} 
      child { node [s_bla] {} 
        child { node [s_bla] {} 
          edge from parent node[left] {$0$} }
        child { node [s_red] {} 
          edge from parent node[right] {$1$} }
        edge from parent node[left] {$0$} }
      child { node [s_red] {} 
        child { node [s_red] {} 
          edge from parent node[left] {$0$} }
        child { node [s_bla] {} 
          edge from parent node[right] {$1$} }
        edge from parent node[right] {$1$} }
      edge from parent node[right] {$1$} } 
    edge from parent node[left] {$1$}
  }
;
\end{tikzpicture}
\end{minipage}
\caption{Prefixes of height $4$ of two decorated trees.}
  \label{fig:dec_tree}
\end{figure}
\end{example}

We use the terminology of \cite{sturm} where Sturmian trees are studied; it is relevant to consider (labeled and decorated) factors occurring in trees.

\begin{definition}
The {\em domain} $\dom(T)$ of a labeled tree $T$ is the set of labels of paths from the root to its nodes. In particular, $\dom(T(L))=L$ for any prefix-closed language $L$ over an ordered (finite) alphabet. The {\em truncation} of a tree at height~$h$ is the restriction of the tree to the domain
$\dom(T)\cap A^{\le h}$.
\end{definition}

Let $L$ be a prefix-closed language over $(A,<)$ and $\mathbf{x}=x_0x_1\cdots$ be an infinite word over some finite alphabet $B$.
From now on, we consider the labeled tree $T(L)$ decorated by $\mathbf{x}$.
(We could use an \textit{ad hoc} notation like $T_\mathbf{x}(L)$ but in any case we only work with decorated trees and it would make the presentation cumbersome.)
For all $n\ge 0$, the $n$th word $w_n$ in $L$ corresponds to the $n$th node of $T(L)$ decorated by $x_n$. Otherwise stated, for the ANS $\mathcal{S}=(L,A,<)$ built on $L$, if $w\in L$, the node corresponding to $w$ in $T(L)$ has decoration $x_{\val_\mathcal{S}(w)}$.

\begin{definition}
Let $w\in L$. We let $T[w]$ denote the subtree of $T$ having $w$ as root. Its domain is $w^{-1}L=\{u \mid wu\in L\}$. We say that $T[w]$ is a {\em suffix} of $T$. 

For any $h\ge 0$, we let $T[w,h]$ denote the {\em factor of height $h$ rooted at $w$}, which is the truncation of $T[w]$ at height~$h$. The {\em prefix of height $h$} of $T$ is the factor $T[\varepsilon,h]$. Two factors $T[w,h]$ and $T[w',h]$ of the same height are {\em equal} if they have the same domain and the same decoration, i.e., $x_{\val_\mathcal{S}(wu)}=x_{\val_\mathcal{S}(w'u)}$ for all $u\in \dom(T[w,h])=\dom(T[w',h])$.
We let $$F_h = \{ T[w,h] \mid w\in L\}$$
denote the set of factors of height $h$ occurring in $T$.
The tree $T$ is \emph{rational} if it has finitely many suffixes.
\end{definition}

Note that, due to Remark~\ref{rem:perT}, with any decoration, even constant, the tree $T(L_{\frac{p}{q}})$ is not rational.

In Figure~\ref{fig:A2L32}, we have depicted the factors of height $2$ occurring in $T(L_{\frac32})$ decorated by $\mathbf{t}$. In Figure~\ref{fig:A2L2}, we have depicted the factors of height $2$ occurring in $T(L_2)$ decorated by the Thue--Morse sequence.
In this second example, except for the prefix of height $2$, observe that a factor of height $2$ is completely determined by the decoration of its root.

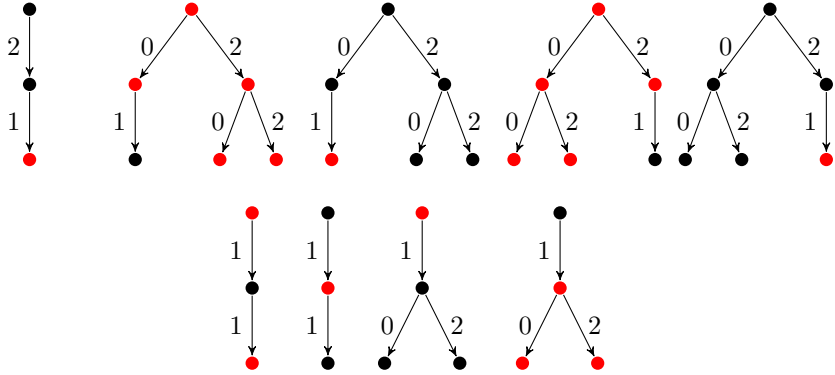
\begin{figure}[h!t]
  \centering
\tikzset{
  s_bla/.style = {circle,fill=black,thick, inner sep=0pt, minimum size=5pt},
  s_red/.style = {circle,black,fill=red,thick, inner sep=0pt, minimum size=5pt}
}
\begin{tikzpicture}[->,>=stealth',level/.style={sibling distance = 2cm/#1},level distance = 1cm]
  \node [s_bla] {} 
    child {node [s_bla] {} 
     child {node [s_red] {} edge from parent node[left] {$1$}} 
      edge from parent node[left] {$2$}}
;
\end{tikzpicture}\hskip.9cm
\begin{tikzpicture}[->,>=stealth',level/.style={sibling distance = 1.5cm/#1},level distance = 1cm]
  \node [s_red] {} 
      child {node [s_red] {} 
      child {node [s_bla] {} edge from parent node[left] {$1$}} 
            edge from parent node[left] {$0$}}
    child {node [s_red] {} 
      child {node [s_red] {} edge from parent node[left] {$0$}} 
      child {node [s_red] {} edge from parent node[right] {$2$}} 
      edge from parent node[right] {$2$}}
;
\end{tikzpicture}
\begin{tikzpicture}[->,>=stealth',level/.style={sibling distance = 1.5cm/#1},level distance = 1cm]
  \node [s_bla] {} 
      child {node [s_bla] {} 
      child {node [s_red] {} edge from parent node[left] {$1$}} 
            edge from parent node[left] {$0$}}
    child {node [s_bla] {} 
      child {node [s_bla] {} edge from parent node[left] {$0$}} 
      child {node [s_bla] {} edge from parent node[right] {$2$}} 
      edge from parent node[right] {$2$}}
;
\end{tikzpicture}
\begin{tikzpicture}[->,>=stealth',level/.style={sibling distance = 1.5cm/#1},level distance = 1cm]
  \node [s_red] {} 
      child {node [s_red] {} 
        child {node [s_red] {} edge from parent node[left] {$0$}} 
        child {node [s_red] {} edge from parent node[right] {$2$}} 
            edge from parent node[left] {$0$}}
    child {node [s_red] {} 
     child {node [s_bla] {} edge from parent node[left] {$1$}} 
      edge from parent node[right] {$2$}}
;
\end{tikzpicture}
\begin{tikzpicture}[->,>=stealth',level/.style={sibling distance = 1.5cm/#1},level distance = 1cm]
  \node [s_bla] {} 
      child {node [s_bla] {} 
        child {node [s_bla] {} edge from parent node[left] {$0$}} 
        child {node [s_bla] {} edge from parent node[right] {$2$}} 
            edge from parent node[left] {$0$}}
    child {node [s_bla] {} 
     child {node [s_red] {} edge from parent node[left] {$1$}} 
      edge from parent node[right] {$2$}}
;
\end{tikzpicture}
\vskip.5cm 

\begin{tikzpicture}[->,>=stealth',level/.style={sibling distance = 2cm/#1},level distance = 1cm]
  \node [s_red] {} 
    child {node [s_bla] {} 
     child {node [s_red] {} edge from parent node[left] {$1$}} 
      edge from parent node[left] {$1$}}
;
\end{tikzpicture}\hskip.5cm
\begin{tikzpicture}[->,>=stealth',level/.style={sibling distance = 2cm/#1},level distance = 1cm]
  \node [s_bla] {} 
    child {node [s_red] {} 
     child {node [s_bla] {} edge from parent node[left] {$1$}} 
      edge from parent node[left] {$1$}}
  ;
\end{tikzpicture}\hskip.5cm
\begin{tikzpicture}[->,>=stealth',level/.style={sibling distance = 2cm/#1},level distance = 1cm]
  \node [s_red] {} 
    child {node [s_bla] {} 
      child {node [s_bla] {} edge from parent node[left] {$0$}}
      child {node [s_bla] {} edge from parent node[right] {$2$}}
      edge from parent node[left] {$1$}}
;
\end{tikzpicture}\hskip.5cm
\begin{tikzpicture}[->,>=stealth',level/.style={sibling distance = 2cm/#1},level distance = 1cm]
  \node [s_bla] {} 
    child {node [s_red] {} 
      child {node [s_red] {} edge from parent node[left] {$0$}}
      child {node [s_red] {} edge from parent node[right] {$2$}}
      edge from parent node[left] {$1$}}
;
\end{tikzpicture}
\caption{The $9$ factors of height $2$ in $T(L_{\frac32})$ decorated by $\mathbf{t}$. The first one is the prefix occurring only once.}
  \label{fig:A2L32}
\end{figure}

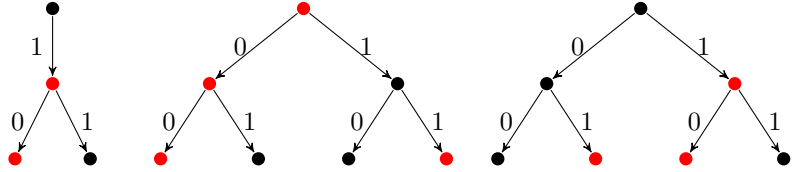
\begin{figure}[h!t]
  \centering
  \tikzset{
  s_bla/.style = {circle,fill=black,thick, inner sep=0pt, minimum size=5pt},
  s_red/.style = {circle,black,fill=red,thick, inner sep=0pt, minimum size=5pt}
}
\begin{tikzpicture}[->,>=stealth',level/.style={sibling distance = 2cm/#1},level distance = 1cm]
  \node [s_bla] {} 
    child {node [s_red] {} 
      child {node [s_red] {} edge from parent node[left] {$0$}}
      child {node [s_bla] {} edge from parent node[right] {$1$}}
      edge from parent node[left] {$1$}}
;
\end{tikzpicture}\hskip.5cm  
\tikzstyle{level 1}=[sibling distance=25mm]
\tikzstyle{level 2}=[sibling distance=13mm]
\begin{tikzpicture}[->,>=stealth',level distance = 1cm]
   \node [s_red] {} 
    child {node [s_red] {} 
      child {node [s_red] {} edge from parent node[left] {$0$}}
      child {node [s_bla] {} edge from parent node[right] {$1$}}
      edge from parent node[left] {$0$}}
     child {node [s_bla] {} 
      child {node [s_bla] {} edge from parent node[left] {$0$}}
      child {node [s_red] {} edge from parent node[right] {$1$}}
      edge from parent node[right] {$1$}}
  ;
\end{tikzpicture}\hskip.5cm
\begin{tikzpicture}[->,>=stealth',level distance = 1cm]
   \node [s_bla] {} 
    child {node [s_bla] {} 
      child {node [s_bla] {} edge from parent node[left] {$0$}}
      child {node [s_red] {} edge from parent node[right] {$1$}}
      edge from parent node[left] {$0$}}
     child {node [s_red] {} 
      child {node [s_red] {} edge from parent node[left] {$0$}}
      child {node [s_bla] {} edge from parent node[right] {$1$}}
      edge from parent node[right] {$1$}}
  ;
\end{tikzpicture}
\caption{The $3$ factors of height $2$ in $T(L_2)$ decorated by the Thue--Morse sequence. The first one is the prefix occurring only once.}
  \label{fig:A2L2}
\end{figure}

Since every factor of height $h$ is the prefix of a factor of height $h+1$, we trivially have $\# F_{h+1}\ge \# F_{h}$. This is quite similar to factors occurring in an infinite word: any factor has at least one extension. In particular, ultimately periodic words are characterized by a bounded factor complexity.

\begin{lemma}\cite[Proposition~1]{sturm}\label{lem:sturmian-tree}
  Let $L$ be a prefix-closed language over $(A,<)$ and let $\mathbf{x}=x_0x_1\cdots$ be an infinite word over some finite alphabet $B$. Consider the labeled tree $T(L)$ decorated by~$\mathbf{x}$. 
  The tree $T(L)$ is rational if and only if $\# F_{h}=\# F_{h+1}$ for some $h\ge 0$. In particular, in that case, $\# F_h = \# F_{h+n}$ for all $n\ge 0$.
\end{lemma}

We can characterize $\mathcal{S}$-automatic sequences built on a prefix-closed regular language $L$ in terms of the decorated tree $T(L)$. For the sake of presentation, we mainly focus on the case of $k$-automatic sequences. The reader can relate our construction to the $k$-kernel of a sequence.
Roughly, each element of the $k$-kernel corresponds to reading one fixed suffix $u$ from each node $w$ of the tree $T(L_k)$.
We have $\val_k(wu)=k^{|u|}\val_k(w)+\val_k(u)$ and an element from the $k$-kernel is a sequence of the form $(x_{k^{|u|} n +\val_k(u)})_{n\ge 0}$.

\begin{theorem} Let $k\ge 2$ be an integer. 
  A sequence $\mathbf{x}$ is $k$-automatic if and only if the labeled tree $T(L_k)$ decorated by $\mathbf{x}$ is rational.
\end{theorem}

\begin{proof}
Let us prove the forward direction.
If $\mathbf{x}$ is $k$-automatic, there exists a DFAO $\mathcal{A}=(Q,q_0,A_k,\delta,\tau)$ producing it when fed with base-$k$ representations of integers.
Let $w\in L_k$ be a non-empty base-$k$ representation.
The suffix $T[w]$ is completely determined by the state $\delta(q_0,w)$.
Indeed, it is a full $k$-ary tree and the decorations are given by $\tau(\delta(q_0,wu))$ for $u$ running through $A_k^*$ in radix order.
For the empty word, however, the suffix $T[\varepsilon]=T$ is decorated by $\tau(\delta(q_0,u))$ for $u$ running through $\{\varepsilon\}\cup\{1,\ldots,k-1\} A_k^*$.
Hence $T(L_k)$ is rational: it has a finite number of suffix trees.

Let us prove the backward direction. 
Assume that the decorated tree $T:=T(L_k)$ is rational.
By definition, the set $Q:=\{T[w] \mid w\in\dom(T)\}$ is finite.
We define a DFAO~$\mathcal{F}$ whose set of states is $Q$
and whose transition function is given by 
$$\forall i\in A_k: \delta(T[w],i)= T[wi].$$
The initial state is given by the tree $T[\varepsilon]=T$ and we set $\delta(T[\varepsilon],0)= T[\varepsilon]$. Finally the output function maps a suffix $T[w]$ to the decoration of its root $w$, that is, $x_{\val_k(w)}$.
If follows that $x_n$ is the output of $\mathcal{F}$ when fed with $\rep_k(n)$. Indeed starting from the initial state $T[\varepsilon]$, we reach the state $T[\rep_k(n)]$ and the output is $x_{\val_k(\rep_k(n))}=x_n$.
\end{proof}

We improve the previous result to ANS with a regular numeration language.

\begin{theorem}
  Let $\mathcal{S}=(L,A,<)$ be an ANS built on a prefix-closed regular language $L$. A sequence $\mathbf{x}$ is $\mathcal{S}$-automatic if and only if the labeled tree $T(L)$ decorated by $\mathbf{x}$ is rational.
\end{theorem}

\begin{proof}
  The proof follows exactly the same lines as for integer base numeration systems.
The only refinement is the following one.
A suffix $T[w]$ of $T(L)$ is determined by $w^{-1}L$ and $\delta(q_0,w)$. Since $L$ is regular, the set $\{w^{-1}L\mid w\in A^*\}$ is finite.
\end{proof}

\subsection{Rational bases} We now turn to rational base numeration systems.  A factor of height $h$ in $T(L_\frac32)$ only depends on the value of its root modulo~$2^h$. This result holds for any rational base numeration system.

\begin{lemma}\cite[Lemme~4.14]{PhD}\label{lem:414}
  Let $w,w'\in L_\frac{p}{q}$ be non-empty words and let $u\in A_p^*$ be a word of length~$h$.
  \begin{itemize}
  \item If $\val_\frac{p}{q}(w)\equiv\val_\frac{p}{q}(w')\bmod{q^h}$, then $u \in w^{-1} L_\frac{p}{q}$ if and only if $u\in (w')^{-1} L_\frac{p}{q}$.
  \item If $u\in (w^{-1} L_\frac{p}{q} \cap (w')^{-1} L_\frac{p}{q})$, then $\val_\frac{p}{q}(w)\equiv\val_\frac{p}{q}(w')\bmod{q^h}$.
  \end{itemize}
\end{lemma}

In the previous lemma, the empty word behaves differently. For a non-empty word $w\in L_\frac{p}{q}$ with $\val_\frac{p}{q}(w)\equiv0\bmod{q^h}$, a word $u\in A_p^h$ not starting with $0$ verifies $u \in \varepsilon^{-1} L_\frac{p}{q}$ if and only if $u\in w^{-1} L_\frac{p}{q}$.
Therefore the prefix of the tree $T(L_\frac{p}{q})$ has to be treated separately.

\begin{lemma}\cite[Corollaire~4.17]{PhD}\label{lem:417}
Every word $u\in A_p^*$ is suffix of a word in $L_\frac{p}{q}$.
\end{lemma}

As a consequence of these lemmas, we obtain the following corollary.

\begin{corollary}\label{cor:partition-h}
For all $h\ge 0$, the set $\{w^{-1}L_\frac{p}{q} \cap A_p^ h\mid w\in A_p^+\}$ is a partition of $A_p^h$ into $q^h$ non-empty languages.
\end{corollary}

Otherwise stated, in the tree $T(L_\frac{p}{q})$ with no decoration or, equivalently with a constant decoration for all nodes, there are $q^h+1$ factors of height $h\ge 1$ (we add $1$ to count the height-$h$ prefix, which has a different shape). 
For instance, if the decorations in Figure~\ref{fig:A2L32} are not taken into account, there are $5=2^2+1$ height-$2$ factors occurring in $T(L_{\frac32})$.

Except for the height-$h$ prefix, each factor of height $h$ is extended in exactly $q$ ways to a factor of height $h+1$. To the first (leftmost) leaf of a factor of height $h$ are attached children corresponding to one of the $q$ words of the periodic labeled signature. To the next leaves on the same level are periodically attached as many nodes as the length of the different words of the signature.
For instance, in the case $\frac{p}{q}=\frac{3}{2}$, the first (leftmost) leaf of a factor of height $h$ becomes a node of degree either $1$ (label $1$) or $2$ (labels $0$ and $2$) to get a factor of height~$h+1$. The next leaves on the same level periodically become nodes of degree $2$ or $1$ accordingly. An example is depicted in Figure~\ref{fig:F23}.

\begin{figure}[h!t]
  \centering
  \begin{subfigure}[b]{.9\linewidth}
   \centering
\tikzset{
  s_bla/.style = {circle,fill=white,thick, draw=black, inner sep=0pt, minimum size=5pt},
  s_red/.style = {circle,black,fill=gray,thick, inner sep=0pt, minimum size=5pt}
}
\begin{tikzpicture}[->,>=stealth',level/.style={sibling distance = 1.5cm/#1},level distance = 1cm]
  \node [s_bla,label=above:$8n+2$] {} 
      child {node [s_bla] {} 
        child {node [s_bla] {} 
          child {node [s_red] {} edge from parent node[left] {$1$} }
          edge from parent node[left] {$1$} }
      edge from parent node[left] {$0$} }
    child {node [s_bla] {} 
      child {node [s_bla] {}
        child {node [s_red] {} edge from parent node[left] {$0$}}
        child {node [s_red] {} edge from parent node[right] {$2$}}
        edge from parent node[left] {$0$}} 
      child {node [s_bla] {}
        child {node [s_red] {} edge from parent node[right] {$1$} }
        edge from parent node[right] {$2$}} 
      edge from parent node[right] {$2$}}
;
\end{tikzpicture}\hskip.5cm
   \begin{tikzpicture}[->,>=stealth',level/.style={sibling distance = 2cm/#1},level distance = 1cm]
  \node [s_bla,label=above:$8n+7$] {} 
    child {node [s_bla] {} 
      child {node [s_bla] {}
        child {node [s_red] {} edge from parent node[left] {$1$}} 
        edge from parent node[left] {$1$}} 
      edge from parent node[left] {$1$}}
  ;
\end{tikzpicture}\hskip.5cm
\begin{tikzpicture}[->,>=stealth',level/.style={sibling distance = 1.5cm/#1},level distance = 1cm]
  \node [s_bla,label=above:$8n+4$] {} 
      child {node [s_bla] {} 
        child {node [s_bla] {}
          child {node [s_red] {} edge from parent node[left] {$1$}} 
          edge from parent node[left] {$0$}} 
        child {node [s_bla] {}
           child {node [s_red] {} edge from parent node[left] {$0$}} 
           child {node [s_red] {} edge from parent node[right] {$2$}} 
           edge from parent node[right] {$2$}} 
            edge from parent node[left] {$0$}}
    child {node [s_bla] {} 
      child {node [s_bla] {}
        child {node [s_red] {} edge from parent node[left] {$1$}} 
        edge from parent node[left] {$1$}} 
      edge from parent node[right] {$2$}}
;
\end{tikzpicture}\hskip.5cm
\begin{tikzpicture}[->,>=stealth',level/.style={sibling distance = 2cm/#1},level distance = 1cm]
  \node [s_bla,label=above:$8n+1$] {} 
    child {node [s_bla] {} 
      child {node [s_bla] {}
        child {node [s_red] {} edge from parent node[left] {$1$}} 
        edge from parent node[left] {$0$}}
      child {node [s_bla] {}
        child {node [s_red] {} edge from parent node[left] {$0$}} 
        child {node [s_red] {} edge from parent node[right] {$2$}} 
        edge from parent node[right] {$2$}}
      edge from parent node[left] {$1$}}
;
\end{tikzpicture}
\caption{The leftmost leaf of each tree is reached by reading a word ending with $1$, so all trees belong to $F^\infty_{3,1}$ (assuming they appear infinitely often).}
\end{subfigure}
\vskip.6cm
 \begin{subfigure}[b]{.9\linewidth}
         \centering
\tikzset{
  s_bla/.style = {circle,fill=white,thick, draw=black, inner sep=0pt, minimum size=5pt},
  s_red/.style = {circle,black,fill=gray,thick, inner sep=0pt, minimum size=5pt}
}
\begin{tikzpicture}[->,>=stealth',level/.style={sibling distance = 1.5cm/#1},level distance = 1cm]
  \node [s_bla,label=above:$8n+6$] {}  
      child {node [s_bla] {} 
        child {node [s_bla] {}
          child {node [s_red] {} edge from parent node[left] {$0$}}
          child {node [s_red] {} edge from parent node[right] {$2$}}
          edge from parent node[left] {$1$}} 
            edge from parent node[left] {$0$}}
    child {node [s_bla] {} 
      child {node [s_bla] {}
        child {node [s_red] {} edge from parent node[left] {$1$} }
        edge from parent node[left] {$0$}} 
      child {node [s_bla] {}
        child {node [s_red] {} edge from parent node[left] {$0$}}
        child {node [s_red] {} edge from parent node[right] {$2$}}
        edge from parent node[right] {$2$}} 
      edge from parent node[right] {$2$}}
;
\end{tikzpicture}\hskip.5cm
  \begin{tikzpicture}[->,>=stealth',level/.style={sibling distance = 2cm/#1},level distance = 1cm]
  \node [s_bla,label=above:$8n+3$] {} 
    child {node [s_bla] {} 
      child {node [s_bla] {}
                child {node [s_red] {} edge from parent node[left] {$0$}} 
                child {node [s_red] {} edge from parent node[right] {$2$}} 
                edge from parent node[left] {$1$}} 
      edge from parent node[left] {$1$}}
  ;
  \end{tikzpicture}\hskip.5cm
  \begin{tikzpicture}[->,>=stealth',level/.style={sibling distance = 1.5cm/#1},level distance = 1cm]
  \node [s_bla,label=above:$8n$] {} 
      child {node [s_bla] {} 
        child {node [s_bla] {}
          child {node [s_red] {} edge from parent node[left] {$0$}} 
          child {node [s_red] {} edge from parent node[right] {$2$}} 
          edge from parent node[left] {$0$}} 
        child {node [s_bla] {}
          child {node [s_red] {} edge from parent node[right] {$1$}} 
          edge from parent node[right] {$2$}} 
            edge from parent node[left] {$0$}}
    child {node [s_bla] {} 
      child {node [s_bla] {}
        child {node [s_red] {} edge from parent node[left] {$0$}} 
        child {node [s_red] {} edge from parent node[right] {$2$}} 
        edge from parent node[left] {$1$}} 
      edge from parent node[right] {$2$}}
;
\end{tikzpicture}\hskip.5cm
\begin{tikzpicture}[->,>=stealth',level/.style={sibling distance = 2cm/#1},level distance = 1cm]
  \node [s_bla,label=above:$8n+5$] {} 
    child {node [s_bla] {} 
      child {node [s_bla] {}
        child {node [s_red] {} edge from parent node[left] {$0$}} 
        child {node [s_red] {} edge from parent node[right] {$2$}} 
        edge from parent node[left] {$0$}}
      child {node [s_bla] {}
        child {node [s_red] {} edge from parent node[right] {$1$}} 
        edge from parent node[right] {$2$}}
      edge from parent node[left] {$1$}}
;
\end{tikzpicture}
\caption{The leftmost leaf of each tree is reached by reading a word ending with $0$, so all trees belong to $F^\infty_{3,0}$ (assuming they appear infinitely often).}
\end{subfigure}
  \caption{For the rational base $\frac{3}{2}$, each factor of height $h=2$ gives $2$ factors of height $h+1=3$.}
  \label{fig:F23}
\end{figure}
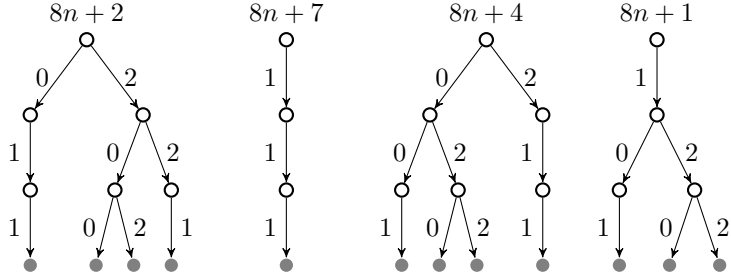
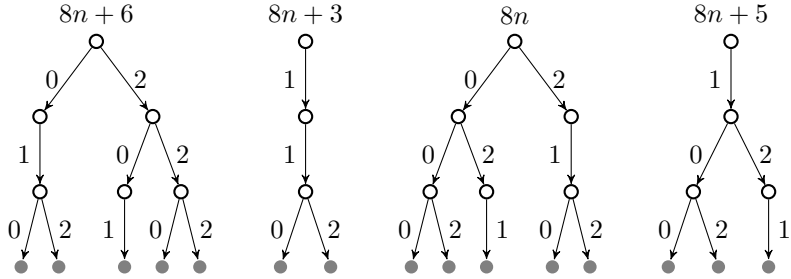

\begin{lemma}
Let $\mathbf{x}$ be a $\frac{p}{q}$-automatic sequence produced by the DFAO $\mathcal{A}=(Q,q_0,A_p,\delta,\tau)$ and let $T(L_{\frac{p}{q}})$ be decorated by $\mathbf{x}$.
For all $h\ge 1$, the number  $\# F_h$ of height-$h$ factors of $T(L_{\frac{p}{q}})$ is bounded by $1+ q^h\cdot\# Q$.
\end{lemma}

\begin{proof}
Let $w\in L_\frac{p}{q}$ be a non-empty base-$\frac{p}{q}$ representation and let $h\ge 1$. 
We claim that the factor $T[w,h]$ is completely determined by the value $\val_\frac{p}{q}(w) \bmod{q^h}$ and the state $\delta(q_0,w)$.
First, from Lemma~\ref{lem:414}, the labeled tree $T[w,h]$ of height $h$ with root $w$ and in particular, its domain, only depends on $\val_\frac{p}{q}(w)$ modulo $q^h$. Indeed, if $w,w'\in L_\frac{p}{q}$ are such that $\val_\frac{p}{q}(w)\equiv\val_\frac{p}{q}(w')\bmod{q^h}$, then 
\[
\dom(T[w,h])=w^{-1}L_\frac{p}{q}\cap A_p^{\le h}=w'^{-1}L_\frac{p}{q}\cap A_p^{\le h}=\dom(T[w',h]).
\]
Second, the decorations of the factor $T[w,h]$ are given by $\tau(\delta(q_0,wu))$ for $u$ running through \linebreak $\dom(T[w,h])=w^{-1}L_\frac{p}{q}\cap A_p^{\le h}$ enumerated in radix order. So the decorations only depend on the state $\delta(q_0,w)$ of $\mathcal{A}$. 
Hence the number of such factors is bounded by $q^h\cdot\# Q$.

Similarly, the height-$h$ prefix $T[\varepsilon,h]$ is decorated by $\tau(\delta(q_0,u))$ for $u$ running through $\dom(T[\varepsilon,h])=L_\frac{p}{q}\cap A_p^{\le h}$. 

Hence $\# F_h$ is bounded by $1+ q^h\cdot\# Q$, for all $h\ge 1$. 
\end{proof}

\begin{definition}
A tree of height $h\ge 0$ has nodes on $h+1$ levels: the {\em level} of a node is its distance to the root. Hence, the root is the only node on level $0$ and the leaves are on level $h$.
\end{definition}

For instance, in Figure~\ref{fig:F23}, each tree of height $3$ has four levels.

\begin{definition}
Let $T$ be a labeled decorated tree and let $h\ge 0$. 
We let $F_h^\infty \subseteq F_h$ denote the set of factors of height $h$ occurring infinitely often in $T$. 
For any suitable letter $a$ in the signature of $T$, we let $F_{h,a}^\infty \subseteq F_h^\infty$ denote the set of factors of height $h$ occurring infinitely often in $T$ such that the label of the edge between the first node on level $h-1$ and its first child is $a$. Otherwise stated, the first word of length $h$ in the domain of the factor ends with $a$.
\end{definition}

\begin{example}
In Figure~\ref{fig:F23}, assuming that they occur infinitely often, the first four trees belong to $F^\infty_{3,1}$ and the last four on the second row belong to $F^\infty_{3,0}$.
\end{example}

Even though the language $L_\frac{p}{q}$ is highly non-regular, we can still handle a subset of $\frac{p}{q}$-automatic sequences.
Roughly, with the next two theorems, we characterize $\frac{p}{q}$-automatic sequences in terms of the extensions that factors of a fixed height occurring infinitely often may have.
As mentioned below, the first result can be notably applied when distinct states of the DFAO producing the sequence have distinct outputs.

In the remaining of the section, we let $(w_0,\ldots,w_{q-1})$ denote the signature of $T(L_{\frac{p}{q}})$.
For all $0\le j \le q-1$ and all $0\le i \le |w_j|-1$, we also let $w_{j,i}$ denote the $i$th letter of $w_j$. For the next statement, recall from Lemma~\ref{lem:414} that for two $\frac{p}{q}$-representations $u,v$, $\val_\frac{p}{q}(u)\equiv\val_\frac{p}{q}(v)\mod{q^h}$ implies that the factors $T[u,h]$ and $T[v,h]$ in $T(L_{\frac{p}{q}})$ have the same domain, i.e., $u^{-1}L_\frac{p}{q} \cap A_p^{\le h}= v^{-1}L_\frac{p}{q} \cap A_p^{\le h}$.

\begin{theorem}\label{thm:weird-condition}
  Let $\mathbf{x}$ be a $\frac{p}{q}$-automatic sequence over a finite alphabet $B$ generated by a DFAO $\mathcal{A}=(Q,q_0,A_p,\delta,\tau : A_p \to B)$ with the following property:
  there exists an integer $h$ such that, for all words $u,v\in L_\frac{p}{q}$ such that $\val_\frac{p}{q}(u)\equiv\val_\frac{p}{q}(v)\mod{q^h}$ and $\delta(q_0,u)\neq \delta(q_0,v)$, there exists a word $w\in u^{-1}L_\frac{p}{q} \cap A_p^{\le h}$ such that $\tau(\delta(q_0,uw))\neq \tau(\delta(q_0,vw))$.
    Then in the tree $T(L_{\frac{p}{q}})$ decorated by $\mathbf{x}$, each factor in $F_{h}^\infty$ can be extended to at most one factor in $F_{h+1,w_{j,0}}^\infty$ for all $0\le j \le q-1$.
\end{theorem}

\begin{proof}
  Consider a factor of height $h$ occurring infinitely often, i.e., there is a sequence $(u_i)_{i\ge 1}$ of words in $L_{\frac{p}{q}}$ such that $T[u_1,h]=T[u_2,h]=T[u_3,h]=\cdots$. From Lemma~\ref{lem:414}, all values $\val_\frac{p}{q}(u_i)$ are congruent to $r$ modulo~$q^h$ for some $0\le r<q^h$. Thus the values of $\val_\frac{p}{q}(u_i)$ modulo $q^{h+1}$ that appear infinitely often take at most $q$ values (among $r,r+q^h,\ldots,r+(q-1)q^h$).
 
 The assumption on the DFAO means that if two words $u,v\in L_\frac{p}{q}$ with $\val_\frac{p}{q}(u)\equiv\val_\frac{p}{q}(v)\mod{q^h}$ are such that $\delta(q_0,u)\neq \delta(q_0,v)$, then $T[u,h]\neq T[v,h]$. Indeed, $u^{-1}L_\frac{p}{q} \cap A_p^{\le h}= v^{-1}L_\frac{p}{q} \cap A_p^{\le h}$ and by assumption, there exists $w\in u^{-1}L_\frac{p}{q} \cap A_p^{\le h}$ such that $\tau(\delta(q_0,uw))\neq \tau(\delta(q_0,vw))$. 
Hence, by contraposition, since $T[u_i,h]=T[u_j,h]$, then $\delta(q_0,u_i)=\delta(q_0,u_j)$.
Consequently, if $T[u_i,h+1]$ and $T[u_j,h+1]$ have the same domain, then $T[u_i,h+1]=T[u_j,h+1]$ because $\delta(q_0,u_iw)=\delta(q_0,u_jw)$ for all words $w\in \dom(T[u_i,h+1])$.

Consequently, no two distinct factors of height $h+1$ occurring infinitely often and having the same domain can have the same prefix of height $h$.
Therefore, each factor $U$ of height $h$ occurring infinitely often gives rise to at most one factor $U'$ of height $h+1$ in every $F_{h+1,w_{j,0}}^\infty$ for $0\le j \le q-1$ ($U$ and the first letter $w_{j,0}$ uniquely determine the domain of $U'$).
\end{proof}

\begin{remark}
  In the case of a $k$-automatic sequence, the assumption of the above theorem is always satisfied. We may apply the usual minimization algorithm about indistinguishable states to the DFAO producing the sequence: two states $r,r'$ are {\em distinguishable} if there exists a word $u$ such that $\tau(\delta(r,u))\neq\tau(\delta(r',u))$. The pairs $\{r,r'\}$ such that $\tau(r)\neq\tau(r')$ are distinguishable (by the empty word). Then proceed recursively: if a not yet distinguished pair $\{r,r'\}$ is such that $\delta(r,a)=s$ and $\delta(r',a)=s'$ for some letter $a$ and an already distinguished pair $\{s,s'\}$, then $\{r,r'\}$ is distinguished. The process stops when no new pair is distinguished and we can merge states that belong to indistinguished pairs. In the resulting DFAO, any two states are distinguished by a word whose length is bounded by the number of states of the DFAO. We can thus apply the above theorem. Notice that for a $k$-automatic sequence, there is no restriction on the word distinguishing states since it belongs to $A_k^*$. The extra requirement that $w\in u^{-1}L_\frac{p}{q} \cap A_p^{\le h}= v^{-1}L_\frac{p}{q} \cap A_p^{\le h}$ is therefore important in the case of rational bases and is not present for base-$k$ numeration systems. 
\end{remark}

\begin{remark}
  For a rational base numeration system, the assumption of the above theorem is always satisfied if the output function $\tau$ is the identity; otherwise stated, if the output function maps distinct states to distinct values. This is for instance the case of our toy example~$\mathbf{t}$.
\end{remark}  

When the output function is not injective,  the situation could be more intricate as shown in the next two examples.

\begin{example}  
Let us now consider a DFAO with a non-injective output function.
Take the cyclic DFAO in Figure~\ref{fig:ex-not-id}.
We show that its output function meets the condition in Theorem~\ref{thm:weird-condition}.
Observe that the states $q_0$ and $q_1$ have the same output. Set $h=2$. Let $u,v\in L_{\frac{3}{2}}$ be two words such that $\val_\frac{3}{2}(u)\equiv\val_\frac{3}{2}(v)\mod{2^2}$. So they can be extended by exactly the same words of length at most $2$ to get base-$\frac{3}{2}$ representations. Assume that $q_0.u=q_i$ and $q_0.v=q_j$ with $i,j\in\{0,1,2\}$ and $i\neq j$. If $\{i,j\}=\{0,1\}$ or $\{1,2\}$, there exists a word $w$ of length $1$ such that $uw,vw\in L_{\frac{3}{2}}$ and the outputs given by $q_0.uw$ and $q_0.vw$ are distinct.  If $\{i,j\}=\{0,2\}$, there exists a word $w$ of length $2$ such that $uw,vw\in L_{\frac{3}{2}}$ and the outputs given by $q_0.uw$ and $q_0.vw$ are distinct.

  \begin{figure}[htb]
\begin{center}
\begin{tikzpicture}
  \tikzstyle{every node}=[shape=circle, fill=none, draw=black,minimum size=20pt, inner sep=2pt]
\node(1) at (0,0) {$q_0|0$};
\node(2) at (2,0) {$q_1|0$};
\node(3) at (1,1.3) {$q_2|1$};
\tikzstyle{every node}=[shape=circle, minimum size=5pt, inner sep=2pt]

\draw [-Latex] (-1,0) to node [above] {} (1);

\draw [-Latex] (1) to [] node [below] {$0,1,2$} (2);
\draw [-Latex] (2) to [] node [right] {$0,1,2$} (3);
\draw [-Latex] (3) to [] node [left] {$0,1,2$} (1);

\end{tikzpicture}
\caption{A DFAO with two distinct outputs but three states.}\label{fig:ex-not-id}
\end{center}
\end{figure}
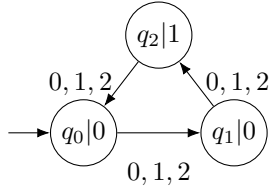
  \end{example}  
  
  \begin{example}  
The condition in Theorem~\ref{thm:weird-condition} might be harder to test than in the previous example.  For instance, take the DFAO depicted in Figure~\ref{fig:counterex} reading base-$\frac{3}{2}$ representations.
The condition in Theorem~\ref{thm:weird-condition} is not met for $h=4$.
For instance the words
$u= 212001220110220$ 
and
$v= 212022000012021$ 
are such that $q_0.u=q_1$, $q_0.v=q_0$ and $u^{-1}L_\frac32\cap A_3^4=v^{-1}L_\frac32\cap A_3^4=\{1111\}$. 
There is no word $w \in u^{-1}L_\frac32\cap A_3^{\le 4}$ such that the outputs of $q_0.uw$ and $q_0.vw$ are distinct. In particular, $T[u,4]=T[v,4]$. 
Furthermore, the reader can observe that the conclusion of Theorem~\ref{thm:weird-condition} does not hold, the latter tree has two extensions as $T[u,5] \neq T[v,5]$ because the states $q_0.u1^40=q_1.0=q_3$ and $q_0.v1^40=q_0.0=q_2$ have distinct outputs (observe $u^{-1}L_\frac32\cap A_3^5=v^{-1}L_\frac32\cap A_3^5=\{11110,11112\}$).

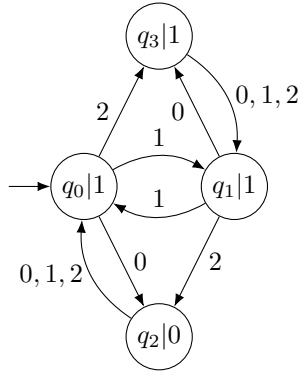
\begin{figure}[htb]
\begin{center}
\begin{tikzpicture}
  \tikzstyle{every node}=[shape=circle, fill=none, draw=black,minimum size=20pt, inner sep=2pt]
\node(1) at (0,0) {$q_0|1$};
\node(2) at (2,0) {$q_1|1$};
\node(3) at (1,-2) {$q_2|0$};
\node(4) at (1,2) {$q_3|1$};
\tikzstyle{every node}=[shape=circle, minimum size=5pt, inner sep=2pt]

\draw [-Latex] (-1,0) to node [above] {} (1);

\draw [-Latex] (1) to [bend left] node [above] {$1$} (2);
\draw [-Latex] (2) to [bend left] node [above] {$1$} (1);
\draw [-Latex] (3) to [bend left] node [left] {$0,1,2$} (1);
\draw [-Latex] (4) to [bend left] node [right] {$0,1,2$} (2);
\draw [-Latex] (1) to [] node [right] {$0$} (3);
\draw [-Latex] (1) to [] node [left] {$2$} (4);
\draw [-Latex] (2) to [] node [right] {$2$} (3);
\draw [-Latex] (2) to [] node [left] {$0$} (4);
\end{tikzpicture}
\caption{A DFAO with two distinct outputs but four states.}\label{fig:counterex}
\end{center}
\end{figure}

We can generalize the above example with the suffix $1^4$.
Let $h\ge 1$ and consider the word $1^h$. From Lemma~\ref{lem:417}, it occurs as a suffix of words in $L_\frac32$. One may thus find words similar to $u$ and $v$ in the above computations.
Actually, $\val_\frac32(u)=591$ and $\val_\frac32(v)=623$ are both congruent to $15=2^4-1$ modulo $2^4$ (so, they can be followed by the suffix $1^4$), and $\val_\frac32(u1^4)$ and $\val_\frac32(v1^4)$ are both even (so, they can be followed by either $0$ or $2$). To have a situation similar to the one with $u$ and $v$ above, we have to look for numbers $n$ which are congruent to $2^h-1$ modulo $2^h$ and such that $$n\left(\frac32\right)^h+\val_\frac32(1^h)=n\left(\frac32\right)^h+\left(\frac32\right)^h-1$$ is an even integer. Numbers of the form $n=(2 j + 1) 2^h - 1$ are convenient. To conclude with this example, a way to show that the DFAO in Figure~\ref{fig:counterex} does not fulfill the condition, is to prove that, for all $h$,  there are two integers $(2 j + 1) 2^h - 1$ and $(2 j' + 1) 2^h - 1$ whose representations lead to respectively to $q_0$ and $q_1$. 
\end{example}

\begin{theorem}
  Let $\mathbf{x}$ be a sequence over a finite alphabet~$B$, and let the tree $T(L_{\frac{p}{q}})$ be decorated by $\mathbf{x}$.
If there exists some $h\ge 0$ such that each factor in $F_{h}^\infty$ can be extended to at most one factor in $F_{h+1,w_{j,0}}^\infty$ for all $0\le j \le q-1$, then $\mathbf{x}$ is $\frac{p}{q}$-automatic.
\end{theorem}

\begin{proof}
 For the sake of readability, write $T=T(L_{\frac{p}{q}})$. The length-$h$ factors of $T$ occurring only a finite number of times appear in a prefix of the tree. Let $t \ge 0$ be the least integer such that 
  all nodes on any level $\ell\ge t$ are roots of a factor in $F_h^{\infty}$.

  We first define a NFA $\mathcal{T}$ in the following way. An illustration that we hope to be helpful is given below in Example~\ref{ex:NFA-construction}. It is made (nodes and edges) of the prefix $T[\varepsilon,t+h-1]$ of height $t+h-1$ and a copy of every element in $F_h^{\infty}$. So the set of states is the union of the nodes of the prefix $T[\varepsilon,t+h-1]$ and the nodes in the trees of $F_h^{\infty}$. 
  Final states are all the nodes of the prefix $T[\varepsilon,t+h-1]$ and the nodes on level exactly $h$ in every element of $F_h^{\infty}$, i.e., the leaves of every element of $F_h^{\infty}$. The unique initial state is the root of the prefix $T[\varepsilon,t+h-1]$. 
  We define the following extra transitions between these elements.
  \begin{itemize}
  \item If a node $m$ on level $t-1$ in the prefix $T[\varepsilon,t+h-1]$ has a child $n$ reached through an arc with label $d$, then in the NFA we add an extra transition with the same label $d$ from $m$ to the root of the element of $F_h^{\infty}$ equal to $T[n,h]$. This is well defined because $n$ on level $t$. 
  \item Let $r$ be the root of an element $T[r,h]$ of $F_h^{\infty}$. Suppose that $r$ has a child $s$ reached through an arc with label $d$.
By assumption the element $T[r,h]$ in $F_h^{\infty}$ can be extended in at most one way to an element $U_c$ in $F_{h+1,c}^{\infty}$ for each $c\in\{w_{0,0},\ldots,w_{q-1,0}\}$. The tree $U_c$ with root $r$ has a subtree of height $h$ with root $rd=s$ denoted by $V_{c,d}\in F_h^{\infty}$ (as depicted in Figure~\ref{fig:rd}; if the extension with $c$ exists, $V_{c,d}$ is unique). In the NFA, we add extra transitions with label $d$ from $r$ to the root of $V_{c,d}$ (there are at most $q$ such trees).
  \end{itemize}

  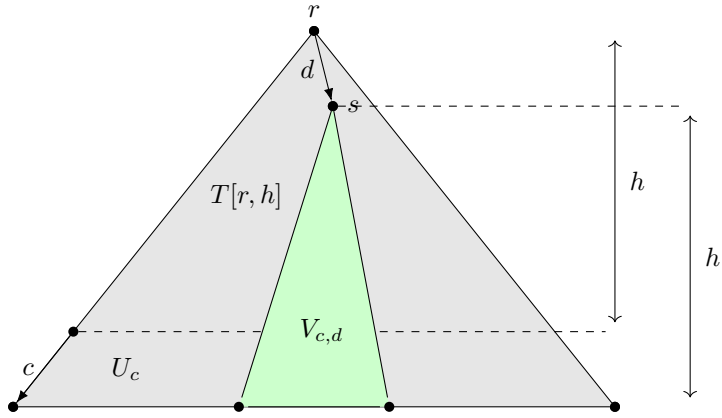
\begin{figure}[h!tb]
    \centering
      \tikzset{
    s_bla/.style = {circle, fill=black, draw, thick, inner sep=1pt, minimum size=3pt}
}
\begin{tikzpicture}
  \node[s_bla] (a1) at (0,0) {};
  \node[s_bla] (a4) at (-4,-5) {};
  \node[s_bla] (a5) at (4,-5) {};
  \fill[fill=gray!20] (a1.center)--(a4.center)--(a5.center);
  \path[draw] (a1)--(a4);
  \path[draw] (a4)--(a5);
  \path[draw] (a5)--(a1);
  \node[s_bla, label=above:$r$] (b1) at (0,0) {};
  \node[s_bla] (b4) at (-4,-5) {};
   \node[s_bla] (b5) at (4,-5) {};
  \node[s_bla] (a3) at (-3.2,-4) {}; 
  \node at (0.25,-1) [s_bla, label=right:$s$] (a2) {};
  \draw [-Latex] (a1) to [] node [left] {$d$} (a2);
  \draw [-Latex] (a3) to [left] node [] {$c$} (a4);
  \node (a6) at (-1,-5) {};
  \node (a7) at (1,-5) {};
  \node (a3') at (3.2,-4) {}; 
  \node (u2) at (4,-4) {};
   \draw [dashed] (a3) -- (u2);

  \fill[fill=green!20] (a2.center)--(a6.center)--(a7.center);
  \path[draw] (a2)--(a6);
  \path[draw] (a6)--(a7);
  \path[draw] (a7)--(a2);
  \node[s_bla] (b6) at (-1,-5) {};
  \node[s_bla] (b7) at (1,-5) {};
  \node[s_bla] (b2) at (0.25,-1) {};
  \node (t1) at (0.1,-4) {$V_{c,d}$};
  \node (t2) at (-.9,-2.2) {$T[r,h]$};
  \node (t3) at (-2.5,-4.5) {$U_c$};
  \node (u1) at (4,0) {};
  \draw [<->] (u1) -- (u2);
  \node (u3) at (4.3,-2) {$h$};
  \node (u1') at (5,-1) {};
  \node (u2') at (5,-5) {};
  \draw [<->] (u1') -- (u2');
  \node (u3') at (5.3,-3) {$h$};
  \draw [dashed] (a2) -- (u1');
\end{tikzpicture}
    \caption{Extension of a tree in $F_h^\infty$.}
    \label{fig:rd}
  \end{figure}
  
We will make use of the following {\em unambiguity property} of $\mathcal{T}$. Every word $u\in L_\frac{p}{q}$ is accepted by $\mathcal{T}$ and there is exactly one successful run for $u$ in $\mathcal{T}$.
If the length of $u\in L_\frac{p}{q}$ is less than $t+h$, there is one successful run and it remains in the prefix $T[\varepsilon,t+h-1]$. If a run uses a transition between a node on level $t-1$ in the prefix $T[\varepsilon,t+h-1]$ and the root of an element in $F_h^\infty$, then the word has to be of length at least $t+h$ to reach a final state by construction.
Now consider a word $u\in L_{\frac{p}{q}}$ of length $t+h+j$ with $j\ge 0$ and write
$$u=u_0\cdots u_{t-1} u_t u_{t+1} \cdots u_{t+h-1} \cdots u_{t+h+j-1}.$$
Reading the prefix $u_0\cdots u_{t-1}$ leads to the root of an element $U$ in $F_h^\infty$.
Assume that this element can be extended in (at least) two ways to a tree of height $h+1$. This means that in $\mathcal{T}$, we have two transitions from the root of $U$ with label $u_t$: one going to the root of some $V_1\in F_{h,c_1}^\infty$ and one going to the root of some $V_2\in F_{h,c_2}^\infty$ with $c_1\neq c_2$ (by assumption, $c_1=c_2$ implies $V_1=V_2$).
Even if $V_1$ and $V_2$ have the same prefix of height~$h-1$, we have $\dom(V_1) \cap \dom(V_2) \cap A_p^h = \emptyset$. This is a consequence of Corollary~\ref{cor:partition-h}: the difference between $\dom(V_1)$ and $\dom(V_2)$ appears precisely on level $h$ where the labeling is periodically $(w_e,w_{e+1},\ldots,w_{q-1},w_0,\ldots,w_{e-1})$ and $(w_f,w_{f+1},\ldots,\ldots,w_{q-1},w_0,\ldots,w_{f-1})$ respectively where $w_e$ (respectively $w_f$) starts with $c_1$ (respectively $c_2$) and the two $q$-tuples of words are a cycle shift of the signature $(w_0,\ldots,w_{q-1})$ of $T$.
So if we non-deterministically make the wrong choice of transition while reading $u_t$, we will not be able to process the letter $u_{t+h}$. The choice of a transition determines the words of length~$h$ that can be read from that point on. The same reasoning occurs for the decision taken at step $t+j$ and the letter processed at step $t+h+j$.


We still have to turn $\mathcal{T}$ into a DFAO producing $\mathbf{x}\in B^\mathbb{N}$. To do so, we determinize $\mathcal{T}$ with the classical subset construction.
Thanks to the unambiguity property of $\mathcal{T}$, if a subset of states obtained during the construction contains final states of $\mathcal{T}$, then they are all decorated by the same letter $b\in B$. The output of this state is thus set to $b$. If a subset of states obtained during the construction contains no final state, then its output is irrelevant (it can be set to any value).   
\end{proof}

\begin{example}\label{ex:NFA-construction}
  Consider the rational base $\frac{3}{2}$. Our aim is to illustrate the above theorem: we have information about factors of a decorated tree $T(L_\frac32)$ --- those occurring infinitely often and those occurring only a finite number of times --- and we want to build the corresponding $\frac32$-automatic sequence.
  Assume that $t=h=1$ and that factors of length~$1$ can be extended as in Figure~\ref{fig:A2L32}. We assume that the last eight trees of height $2$ occur infinitely often. Hence their four prefixes of height~$1$ have exactly two extensions. We assume that the prefix given by the first tree in Figure~\ref{fig:A2L32} occurs only once.

From this, we build the NFA $\mathcal{T}$ depicted in Figure~\ref{fig:NFA}. The prefix tree of height $t+h-1=1$ is depicted on the left and its root is the initial state. The single word $2$ of length $1$ is accepted by a run staying in this tree. Then, are represented the four trees of $F_1^\infty$. Their respective leaves are final states. Finally, we have to inspect Figure~\ref{fig:A2L32} to determine the transitions connecting roots of these trees.
For instance, let us focus on state $7$ in Figure~\ref{fig:NFA}.
On Figure~\ref{fig:A2L32}, the corresponding tree can be extended in two ways: the second and the fourth trees on the first row.
In the first of these trees, the tree hanging to the child $0$ (respectively $2$) of the root corresponds to state $5$ (respectively $7$).
Hence, there is a transition of label $0$ (respectively $2$) from $7$ to $5$ (respectively $7$) in Figure~\ref{fig:NFA}.
Similarly, the second tree gives the extra transitions of label $0$ from $7$ to $7$ and of label $2$ from $7$ to $5$.

\begin{figure}[h!t]
  \centering
\tikzset{
  s_bla/.style = {circle,fill=gray, draw, thick, inner sep=1pt, minimum size=5pt},
  s_red/.style = {circle,fill=red, draw, thick, inner sep=1pt, minimum size=5pt}
}
\begin{tikzpicture}[->,>=stealth',level/.style={sibling distance = 1.5cm/#1},level distance = 1cm]
 \node at (0,0) [s_bla] (a1) {$q_0$};
 \node at (0,-1) [s_bla] (b1) {$q_1$};
 \draw (a1) -> (b1) node[midway, left] {$2$};
 \node[draw=none] (s) [left of=a1] {};
 \draw (s) -> (a1) node {};
 \node[draw=none] at (0,-1.7) (b1') {};
 \draw (b1) -> (b1') node {};
 \node[draw=none] at (0,0.7) (b1'') {};
 \draw (a1) -> (b1'') node {};
  
  \node at (3,0) [s_bla] (a2) {0};
 \node at (3,-1) [s_red] (b2) {1};
 \draw (a2) -> (b2) node[midway, left] {$1$};
 \node[draw=none] at (3,-1.7) (b2') {};
 \draw (b2) -> (b2') node {};
 
  \node at (6,0) [s_bla] (a3) {2};
  \node at (5.5,-1) [s_bla] (b3) {3};
   \node at (6.5,-1) [s_bla] (c3) {4};
 \draw (a3) -> (b3) node[midway, left] {$0$};
 \draw (a3) -> (c3) node[midway, right] {$2$};
  \node[draw=none] at (5.5,-1.7) (b3') {};
  \draw (b3) -> (b3') node {};
    \node[draw=none] at (6.5,-1.7) (c3') {};
    \draw (c3) -> (c3') node {};
    
   \node at (3,-3) [s_red] (a4) {5};
 \node at (3,-4) [s_bla] (b4) {6};
 \draw (a4) -> (b4) node[midway, left] {$1$};
  \node[draw=none] at (3,-4.7) (b4') {};
  \draw (b4) -> (b4') node {};

  \node at (6,-3) [s_red] (a5) {7};
  \node at (5.5,-4) [s_red] (b5) {8};
   \node at (6.5,-4) [s_red] (c5) {9};
 \draw (a5) -> (b5) node[midway, left] {$0$};
 \draw (a5) -> (c5) node[midway, right] {$2$};
  \node[draw=none] at (5.5,-4.7) (b5') {};
  \draw (b5) -> (b5') node {};
    \node[draw=none] at (6.5,-4.7) (c5') {};
 \draw (c5) -> (c5') node {};
 
 \draw (a1) -> (a2) node[midway, above] {$2$};
 \draw (a2) -> (a5) node[near start, above] {$1$};
 \draw (a3) -> (a2) node[midway, above] {$0,2$};
  \draw (a5) -> (a4) node[midway, above] {$0,2$};
  \path
  (a3) edge [loop above] node {$0,2$} (a3)
  (a5) edge [loop above] node {$0,2$} (a5)
  (a4) edge[bend left] node [near start, left] {$1$} (a3)
  (a4) edge[bend left=50,<->] node [midway, left] {$1$} (a2);
\end{tikzpicture}
\caption{A NFA~$\mathcal{T}$.}
  \label{fig:NFA}
\end{figure}
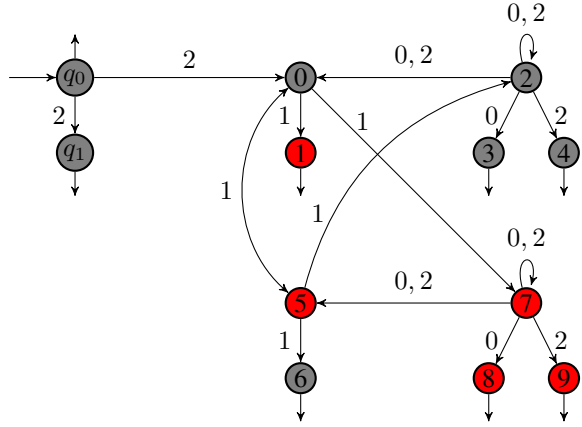

Take the word $210\in L_{\frac{3}{2}}$. Starting from $q_0$, the only successful run is
  $q_0\stackrel{2}{\longrightarrow}0\stackrel{1}{\longrightarrow}7\stackrel{0}{\longrightarrow}8$. If we had reached $0$ with $q_0\stackrel{2}{\longrightarrow}0$ and chose the other transition of label $1$, we would have the run $q_0\stackrel{2}{\longrightarrow}0\stackrel{1}{\longrightarrow}5$, but from state $5$ there is no transition with label $0$. The successful runs of the first few words in $L_{\frac{3}{2}}$ are given below:
  $$\begin{array}{r|l}
      \varepsilon & q_0\\
  2& q_0 \to q_1\\
  21& q_0 \to 0 \to 1\\
  210& q_0 \to 0 \to 7 \to 8\\
  212& q_0 \to 0 \to 7 \to 9\\
  2101& q_0 \to 0 \to 7 \to 5 \to 6\\
  2120& q_0 \to 0 \to 7 \to 7 \to 8\\
  2122& q_0 \to 0 \to 7 \to 7 \to 9\\
  21011& q_0 \to 0 \to 7 \to 5 \to 0 \to 1\\
  21200& q_0 \to 0 \to 7 \to 7 \to 7 \to 8\\
  21202& q_0 \to 0 \to 7 \to 7 \to 7 \to 9\\
  21221& q_0 \to 0 \to 7 \to 7 \to 5 \to 6\\
  \end{array}$$

We may now determinize this NFA $\mathcal{T}$.  We apply the classical subset construction to get a DFAO. If a subset of states contains a final state of $\mathcal{T}$ from $\{1,8,9\}$ (respectively $\{q_0,q_1,3,4,6\}$), the corresponding decoration being $1$ (respectively $0$), the output for this state is $1$ (respectively $0$). Indeed, as explained in the proof, a subset of states of $\mathcal{T}$ obtained during the determinization algorithm cannot contain states with two distinct decorations. After determinization, we obtain the (minimal) DFAO depicted in Figure~\ref{fig:determinize}. In the latter figure, we have not set any output for state $2$ because it corresponds to a subset of states in $\mathcal{T}$ which does not contain any final state. Otherwise stated, that particular output is irrelevant as no valid representation will end up in that state.

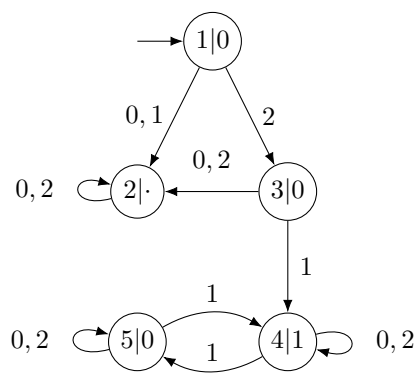
\begin{figure}[htb]
\begin{center}
\begin{tikzpicture}
  \tikzstyle{every node}=[shape=circle, fill=none, draw=black,minimum size=20pt, inner sep=2pt]
\node(1) at (1,2) {$1|0$};
\node(2) at (0,0) {$2|\cdot$};
\node(3) at (2,0) {$3|0$};
\node(4) at (2,-2) {$4|1$};
\node(5) at (0,-2) {$5|0$};
\tikzstyle{every node}=[shape=circle, minimum size=5pt, inner sep=2pt]

\draw [-Latex] (0,2) to node [above] {} (1);

\draw [-Latex] (2) to [loop left] node [left=.2] {$0,2$} (2);
\draw [-Latex] (5) to [loop left] node [left=.2] {$0,2$} (5);
\draw [-Latex] (4) to [loop right] node [right=.2] {$0,2$} (4);
\draw [-Latex] (1) to [] node [left] {$0,1$} (2);
\draw [-Latex] (1) to [] node [right] {$2$} (3);
\draw [-Latex] (3) to [] node [above] {$0,2$} (2);
\draw [-Latex] (3) to [] node [right] {$1$} (4);
\draw [-Latex] (4) to [bend left] node [above] {$1$} (5);
\draw [-Latex] (5) to [bend left] node [above] {$1$} (4);
\end{tikzpicture}
\caption{Determinization of $\mathcal{T}$.}\label{fig:determinize}
\end{center}
\end{figure}
\end{example}

\section{Recognizable sets and stability properties}\label{sec:stability}

In this short section, our aim is to present some direct closure properties of automatic sequences in ANS built on tree languages. These statements should not surprise the reader used to constructions of automata and automatic sequences.

In \cite{Marsault}, a subset $X$ of $N_{\frac{p}{q}}$ is said to be {\em $\frac{p}{q}$-recognizable} if there exists a DFA over $A_p$ accepting a language $L$ such that $\val_{\frac{p}{q}}(L)=X$. Since $L_{\frac{p}{q}}$ is not regular, the set $\mathbb{N}$ is not $\frac{p}{q}$-recognizable.

\begin{prop}
  A sequence $\mathbf{x}=x_0x_1\cdots$ over $A$ is $\frac{p}{q}$-automatic if and only if, for every $a\in A$, there exists a $\frac{p}{q}$-recognizable set $R_a$ such that $\{ i\in\mathbb{N} : x_i=a\}=R_a\cap\mathbb{N}$.
\end{prop}

  \begin{proof}
In the DFAO producing the sequence, consider as final the states having output $a$. The accepted set is $R_a$.
  \end{proof}

For $k$-automatic sequences, the above result can also be expressed in terms of fibers (see, for instance, \cite[Lemma~5.2.6]{Allouche--Shallit-2003}). The {\em $\frac{p}{q}$-fiber} of an infinite sequence $\mathbf{x}$ is the language $I_{\frac{p}{q}}(\mathbf{x},a)=\{\rep_{\frac{p}{q}}(i) : i\in\mathbb{N} \text{ and } x_i=a\}$. A sequence $\mathbf{x}=x_0x_1\cdots$ over $A$ is $\frac{p}{q}$-automatic if and only if, for every $a\in A$, there exists a regular language $S_a$ such that $I_{\frac{p}{q}}(\mathbf{x},a)=S_a\cap L_{\frac{p}{q}}$.

We can verbatim take several robustness or closure properties of automatic sequences. They use classical constructions of automata such as reversal or compositions.

\begin{prop}
Let $\mathcal{S}$ be an abstract
numeration system built on a tree language with a purely periodic labeled signature. The set of $\mathcal{S}$-automatic sequences is stable under finite modifications.
\end{prop}

\begin{proof}
One has to adapt the DFAO to take into account those finite modifications. Suppose that these modifications occur for representations of length at most $\ell$. Then the DFAO can have a tree-like structure for words of length up to $\ell$ and we enter the original DFAO after passing through this structure encoding the modifications.
\end{proof}

\begin{prop}
Let $\mathcal{S}$ be an abstract
numeration system built on a tree language with a purely periodic labeled signature. The set of $\mathcal{S}$-automatic sequences is stable under codings.
\end{prop}

Automatic sequences can be produced by reading least significant digits first.
Simply adapt the corresponding result in \cite{Maes}.

\begin{prop}
Let $\mathcal{S}=(L,A,<)$ be an abstract
numeration system built on a tree language with a purely periodic labeled signature.
A sequence $\mathbf{x}$ is $\mathcal{S}$-automatic if and only if there exists a DFAO $(Q,q_0,A,\delta,\tau)$ such that, for all $n\ge 0$, $x_n=\tau(\delta(q_0,(\rep_{\mathcal{S}}(n))^R))$.
\end{prop}

Adding leading zeroes does not affect automaticity.
Simply adapt the proof of \cite[Theorem~5.2.1]{Allouche--Shallit-2003}.

\begin{prop}
A sequence $\mathbf{x}$ is $\frac{p}{q}$-automatic if and only if there exists a DFAO $(Q,q_0,A_p,\delta,\tau)$ such that, for all $n\ge 0$ and all $j\ge 0$, $x_n=\tau(\delta(q_0,0^j\rep_{\frac{p}{q}}(n)))$.
\end{prop}

For any finite alphabet $D\subset\mathbb{Z}$ of digits, we let $\chi_D$ denote the \emph{digit-conversion} map defined as follows: for all $u\in D^*$ such that $\val_{\frac{p}{q}}(u)\in\mathbb{N}$, $\chi_D(u)$ is the unique word $v\in L_{\frac{p}{q}}$ such that $\val_{\frac{p}{q}}(u)=\val_{\frac{p}{q}}(v)$. In \cite{Akiyama--Frougny-Sakarovitch-2008}, it is shown that $\chi_D$ can be realized by a finite letter-to-letter right transducer. As a consequence of this result, multiplication by a constant $a\ge 1$ is realized by a finite letter-to-letter right transducer. Indeed take a word $u=u_0\cdots u_t\in L_{\frac{p}{q}}$ and consider the alphabet $D=\{0,a,2a,\ldots, (p-1)a\}$. Feed the transducer realizing $\chi_D$ with $a u_t$, \ldots, $a u_0$. The output is the base-$\frac{p}{q}$ representation of $a\cdot \val_{\frac{p}{q}}(u)$. Similarly, translation by a constant $b\ge 0$  is realized by a finite letter-to-letter right transducer. Consider the alphabet $D'=\{0,\ldots,p+b-1\}$. Feed the transducer realizing $\chi_{D'}$ with $(u_t+b)$, $u_{t-1}$, \ldots, $u_0$. The output is the base-$\frac{p}{q}$ representation of $\val_{\frac{p}{q}}(u)+b$. Combining these results with the DFAO producing a $\frac{p}{q}$-automatic sequence, we get the following result.

 \begin{corollary}
   Let $a\ge 1, b\ge 0$ be integers.
   If a sequence $\mathbf{x}$ is $\frac{p}{q}$-automatic, then the sequence $(x_{an+b})_{n\ge 0}$ is also $\frac{p}{q}$-automatic.
 \end{corollary}

\begin{remark}
Ultimately periodic sequences are $k$-automatic for any integer $k\ge 2$~\cite[Theorem 5.4.2]{Allouche--Shallit-2003}.
They are also $\mathcal{S}$-automatic for any abstract numeration system $\mathcal{S}$ based on a regular language~\cite{LR}.
In general, this is not the case for $\frac{p}{q}$-automaticity: the characteristic sequence of multiples of $q$ is not $\frac{p}{q}$-automatic~\cite[Proposition 5.39]{Marsault}.
Nevertheless when the period length of an ultimately periodic sequence is coprime with $q$, then the sequence is $\frac{p}{q}$-automatic~\cite[Théorème 5.34]{Marsault}.
\end{remark}

\acknowledgements{We thank the anonymous referees for their suggestions greatly improving our presentation.}

\end{document}